\RequirePackage{etoolbox}
\newtoggle{conference}

\togglefalse{conference} 

\newcommand{\inConference}[1]{\iftoggle{conference}{#1}{}} 
\newcommand{\inArxiv}[1]{\iftoggle{conference}{}{#1}}  

\inArxiv{\documentclass[11pt]{article}}
\inConference{
	\documentclass[twoside,leqno,twocolumn]{article}  
	\usepackage{ltexpprt}

	\newcommand{\qedhere}{}
}
\usepackage{graphicx}
\usepackage{latexsym}
\usepackage{amssymb}
\usepackage{amsmath}
\inArxiv{
	\usepackage{amsthm}
	\usepackage[margin=1in,paper=letterpaper]{geometry}
	\usepackage[ruled,vlined,linesnumbered]{algorithm2e}
}
\inConference{
	\usepackage[ruled,vlined,linesnumbered,norelsize]{algorithm2e}
}
\usepackage{paralist}
\usepackage{nicefrac}
\usepackage{multirow}
\usepackage{tikz}
\usepackage{textcomp}
\inArxiv{\usepackage{hyperref}}

\usetikzlibrary{arrows}
\tikzset{
	node/.style      = {draw, circle, node distance = 2cm},
}

\begin{document}

\newcommand {\ignore} [1] {}

\inArxiv{
	\newtheorem{theorem}{Theorem}[section]
	\newtheorem{lemma}[theorem]{Lemma}
	\newtheorem{fact}[theorem]{Fact}
	\newtheorem{corollary}[theorem]{Corollary}
	\newtheorem{proposition}[theorem]{Proposition}
	\newtheorem{definition}{Definition}[section]
	\newtheorem{observation}[theorem]{Observation}
	\newtheorem{claim}[theorem]{Claim}
	\newtheorem{assumption}[theorem]{Assumption}
	\newtheorem{notation}[theorem]{Notation}
	\newtheorem{reduction}{Reduction}
}
\inConference
{
	\newtheorem{observation}{Observation}
}

\def \ee   {\varepsilon}

\def \PP   {{\cal P}}
\def \QQ   {{\cal Q}}
\def \DD   {{\cal D}}
\def \NN   {{\cal N}}
\def \GG   {{\cal G}}
\def \II   {{\cal I}}
\def \MM   {{\cal M}}
\def \VV   {{\cal V}}
\def \EE   {{\cal E}}

\newcommand{\ie}{{\it i.e.}}
\newcommand{\eg}{{\it e.g.}}
\newcommand{\characteristic}{{\mathbf{1}}}
\newcommand{\RSet}{{\mathtt{R}}}
\newcommand{\Win}{{w_{\mathrm{in}}}}
\newcommand{\Wout}{{w_{\mathrm{out}}}}

\newcommand{\create}[1]
{
	\expandafter\newcommand\csname #1\endcsname{{the \csname core#1\endcsname}}
	\expandafter\newcommand\csname S#1\endcsname{{The \csname core#1\endcsname}}
	\expandafter\newcommand\csname G#1\endcsname{{the general \csname core#1\endcsname}}
}
\newcommand{\coreUM}{{\texttt{unconst\-rained-model}}} \create{UM}
\newcommand{\coreDM}{{\texttt{dicut-model}}} \create{DM}
\newcommand{\coreCM}{{\texttt{cardi\-nal\-ity-mo\-del}}} \create{CM}
\newcommand{\email}[1]{{\inArxiv{\href{mailto:#1}{\texttt{#1}}}\inConference{\texttt{#1}}}}

\pagenumbering{arabic}

\title{Online Submodular Maximization with Preemption}

\author{
 Niv Buchbinder\thanks{Statistics and Operations Research Dept., Tel Aviv University, Israel. E-mail: \email{niv.buchbinder@gmail.com}. Research supported by ISF grant 954/11 and BSF grant 2010426.}
 \and
 Moran Feldman\thanks{School of Computer and Communication Sciences, EPFL, Switzerland. E-mail: \email{moran.feldman@epfl.ch}. Research supported in part by ERC Starting Grant 335288-OptApprox.}
 \and
 Roy Schwartz\thanks{Dept. of Computer Science, Princeton University, Princeton, NJ. E-mail: \email{roysch@cs.princeton.edu}.}
}

\inConference{\date{}}
\maketitle

\begin{abstract}
\expandafter\csname\inConference{small}\endcsname
\inConference{\baselineskip=9pt}
Submodular function maximization has been studied extensively in recent years under various constraints and models. The problem plays a major role in various disciplines.
We study a natural online variant of this problem in which elements arrive one-by-one and the algorithm has to maintain a solution obeying certain constraints at all times. Upon arrival of an element, the algorithm has to decide whether to accept the element into its solution and may preempt previously chosen elements. The goal is to maximize a submodular function over the set of elements in the solution.

We study two special cases of this general problem and derive upper and lower bounds on the competitive ratio.  
Specifically, we design a $1/e$-competitive algorithm for the unconstrained case in which the algorithm may hold any subset of the elements, and constant competitive ratio algorithms for the case where the algorithm may hold at most $k$ elements in its solution.

\end{abstract} 

\inArxiv{
	\pagenumbering{Alph}
	\thispagestyle{empty}
	\newpage
	\setcounter{page}{1}
	\pagenumbering{arabic}
}

\section{Introduction}

Submodular function maximization has been studied extensively in recent years under various constraints and models.
Submodular functions, which capture well the intuitive idea of diminishing returns, play a major role in various disciplines, including combinatorics, optimization, economics, information theory, operations research, algorithmic game theory and machine learning. Submodular maximization captures well known combinatorial optimization problems such as: Max-Cut~\cite{GW95,H01,K72,KKMO07,TSSW00}, Max-DiCut~\cite{FG95,GW95,HZ01},
Generalized Assignment~\cite{CK05,CKR06,FV06,FGMS06}, and Max-Facility-Location~\cite{AS99,CFN77a,CFN77b}. Additionally, one can find submodular maximization problems in many other settings. In machine learning, \inArxiv{maximization of submodular functions}\inConference{submodular maximization} has been used for document summarization~\cite{LB10,LB11}, sensor placement~\cite{KSG08,KG05,KLGVF08}, and information gathering~\cite{KG07}. In algorithmic game theory, calculating market expansion~\cite{DRS09} and computing core values of certain types of games~\cite{SU13} are two examples where the problem can be reduced to submodular maximization.

It is natural to consider online variants of submodular maximization, such as the following setting. There is an unknown ground set of elements $\NN = \{u_1, u_2, \dotsc, u_n\}$, a non-negative submodular function $f : 2^\NN \to \mathbb{R}^+$ and perhaps a constraint determining the feasible sets of elements that may be chosen. The elements of $\NN$ arrive one by one in an online fashion. Upon arrival, the online algorithm must decide whether to accept each revealed element into its solution and this decision is irrevocable. As in the offline case, the algorithm has access to the function $f$ via a value oracle, but in the online case it may only query subsets of elements that have already been revealed. For a concrete example, consider a soccer team manager preparing his team for the 2014 World Cup and hiring $k~(=11)$ players. Occasionally, the manager gets an additional application from a new potential player, and may recruit him to the team. The goal is to maximize the quality of the team. The quality of the team is indeed a (non-monotone) submodular function of the chosen players.

Although natural, it is not difficult to see that no algorithm has a constant competitive ratio for the above na\"{\i}ve model even with a simple cardinality constraint. Therefore, in order to obtain meaningful results, one must relax the model. One possible relaxation is to consider a random arrival model in which elements arrive at a random order \cite{BHZ10,FNS11c,GRST10}. This relaxation leads to ``secretary type" algorithms. We propose here a different approach that still allows for an adversarial arrival. Here, we allow the algorithm to reject (preempt) previously accepted elements. Preemption appears as part of many online models (see, \eg, \cite{AA99,EW06,ELSW13,F14}). The use of preemption in our model is mathematically elegant and fits well with natural scenarios. For example, it is natural to assume that a team manager may release one of the players (while hiring another instead) if it benefits the quality of the team.

\subsection{Our Results\inConference{.}}

We study online submodular maximization problems with two simple constraints on the feasible sets. The unconstrained case in which no restrictions are inflicted on the sets of elements that \inArxiv{the algorithm may choose}\inConference{can be chosen}, and the cardinality constraint in which the online algorithm may hold at any time at most $k$ elements. We provide positive results as well as hardness results.
Our hardness results apply both to polynomial and non-polynomial time algorithms, but all our algorithms are polynomial (in some\inArxiv{ of them}, a polynomial time implementation losses an $\ee$ in the competitive ratio).

\paragraph{The unconstrained case:}
The first special case we consider is, arguably, the most basic case in which there is no constraint on the sets of elements that may be chosen. In this case the problem is trivial if the submodular function $f$ is monotone\footnote{A set function $f : 2^\NN \to \mathbb{R}$ is monotone if $A \subseteq B \subseteq \NN$ implies $f(A) \leq f(B)$.} as the algorithm may take all elements. Therefore, we only consider non-monotone functions $f$. One simple candidate algorithm is an algorithm that selects each revealed element independently with probability $1/2$. Feige et al.~\cite{FMV11} proved that this algorithm, which does not use preemption, is $1/4$-competitive. We prove that this is essentially optimal without using the additional power of preemption.

\begin{theorem} \label{th:unconstrained_no_preemption}
An online algorithm using no preemption cannot be $(1/4 + \ee)$-competitive for any constant $\ee > 0$, even if $f$ is guaranteed to be a cut function of a weighted directed graph.
\end{theorem}

On the positive side we prove that preemption is beneficial for {\UM}.

\begin{theorem} \label{th:unconsrained_positive}
There \inArxiv{exists}\inConference{is} a $1/e$-competitive algorithm for {\UM}, which can be implemented in polynomial time at the cost of an $\ee$ loss in the competitive ratio (for an arbitrary small constant $\ee > 0$).
\end{theorem}

A special case of {\UM} that we study is {\DM}. In this model every set is feasible and the objective function $f$ is a cut function of a weighted directed graph $G = (V, A)$ having $\NN \subseteq V$ (\ie, a subset of the nodes of $G$ form the ground set). We assume an algorithm for {\DM} knows $V$ (but not $\NN$) and can query the weight (and existence) of every arc leaving a \emph{revealed} node. Observe that the algorithm can in fact calculate the value of $f$ for every set of revealed elements using this information and never needs to query $f$ via the oracle.

{\SDM} can be viewed as an online model of the well-known Max-DiCut problem (see Section~\ref{ssc:related_work} for a discussion of another online model of Max-DiCut). Additionally, since {\DM} is a special case of {\UM} (with more power for the algorithm), it inherits the positive result given by Theorem~\ref{th:unconsrained_positive}. The next theorem gives a stronger result.

\begin{theorem} \label{th:dicut_positive}
There exists a polynomial time $0.438$-competitive algorithm for {\DM}.
\end{theorem}

Theorem~\ref{th:dicut_positive} is proved by showing that an offline algorithm suggested by~\cite{FJ10} can be implemented under {\DM}. We complement Theorems~\ref{th:unconsrained_positive} and~\ref{th:dicut_positive} by the following theorem which gives hardness results for {\DM} (and thus, also for {\UM}).

\begin{theorem} \label{th:dicut_negative}
\inArxiv{No}\inConference{The competitive ratio of no} randomized (deterministic) algorithm for {\DM} \inArxiv{has a competitive ratio}\inConference{is} better than $4/5$ ($\frac{5 - \sqrt{17}}{2} \approx 0.438$).\footnote{Theorem~\ref{th:dicut_negative} holds even if we allow the algorithm access to all arcs of $G$, including arcs leaving unrevealed elements.}
\end{theorem}

Notice that for polynomial time algorithms a hardness result of $1/2$ proved by~\cite{FMV11} for offline algorithms extends immediately to {\UM}. For {\DM} there exists a polynomial time $0.874$-approximation offline algorithm~\cite{LLZ02}, thus, polynomial time offline algorithms are strictly stronger than online algorithms in this model.

\paragraph{Cardinality constraint:}
The second case we consider is {\CM} in which a set is feasible if and only if its size is at most $k$, for some parameter $0 \leq k \leq n$. Our positive results for {\CM} are summarized by the following theorem.

\begin{theorem} \label{th:cardinality_positive}
There exists a randomized ($56/627 \approx 0.0893$)-competitive algorithm for {\CM}, which can be implemented in polynomial time at the cost of an $\ee$ loss in the competitive ratio (for an arbitrary small constant $\ee > 0$). Moreover, if the objective function is restricted to be monotone, then there exists a polynomial time deterministic ($1/4$)-approximation algorithm for this model.
\end{theorem}

The second part of Theorem~\ref{th:cardinality_positive} follows from the \inArxiv{works of}\inConference{results given by} Ashwinkumar~\cite{A11} and Chakrabarti and Kale~\cite{CK13}. More specifically, these works present a streaming $1/4$-competitive algorithm for maximizing a monotone submodular function subject to a cardinality constraint,\footnote{In fact the algorithm of~\cite{A11} and~\cite{CK13} is $(4m)^{-1}$-competitive for the more general constraint of $m$-matroids intersection.} and this algorithm can be implemented also under {\CM}. We describe a different $1/4$-competitive algorithm for monotone objectives under {\CM}, and then use additional ideas to prove the first part of Theorem~\ref{th:cardinality_positive}.

It is interesting to note that both algorithms guaranteed by Theorem~\ref{th:cardinality_positive} can be implemented in the streaming model of Chakrabarti and Kale~\cite{CK13}. Thus, we give also the first algorithm for maximizing a general non-negative submodular function subject to a cardinality constraint under this streaming model.

On the negative side, notice that {\CM} generalizes {\UM} (by setting $k = n$). Hence, both hardness results given by Theorem~\ref{th:dicut_negative} extend to {\CM}. The following theorem gives a few additional hardness results for this model.

\begin{theorem} \label{th:cardinality_negative}
No algorithm for {\CM} is ($1/2 + \ee$)-competitive for any constant $\ee > 0$. Moreover, even if the objective function is restricted to be monotone, no randomized (deterministic) algorithm for {\CM} is ($3/4 + \ee$)-competitive (($1/2 + \ee)$-competitive).
\end{theorem}

Notice that polynomial time hardness results of $0.491$ and $1 - 1/e$ for {\CM} and {\CM} with a monotone objective follow from~\cite{GV11} and~\cite{NW78}, respectively. All the results for {\CM} are summarized in Table~\ref{th:cardinality_results}.

\begin{table*}
\inArxiv{\small}
\begin{tabular}{|l|cr|cr|cr|cr|}
\cline{2-9}
\multicolumn{1}{c|}{} & \multicolumn{4}{|c|}{\textbf{Monotone Objective}} & \multicolumn{4}{|c|}{\textbf{General Objective}}\\
\cline{2-9}
\multicolumn{1}{c|}{} & \multicolumn{2}{|c|}{\textbf{Deterministic}} & \multicolumn{2}{|c|}{\textbf{Randomized}} & \multicolumn{2}{|c|}{\textbf{Deterministic}} & \multicolumn{2}{|c|}{\textbf{Randomized}} \\
\hline
\rule[-0.25cm]{0pt}{0.7cm}
Algorithm &	$1/4$																		& \cite{CK13,A11} 								& $1/4$									 	 &  \cite{CK13,A11}
					& \multicolumn{2}{|c|}{$-$}									 																& $\frac{56}{627} \approx 0.0893$ & (\ref{th:cardinality_positive}) \\
\hline
\rule[-0.25cm]{0pt}{0.7cm}
Hardness	& $1/2 + \ee$ 														& (\ref{th:cardinality_negative}) & $3/4 + \ee$ 						 & (\ref{th:cardinality_negative})
				  & $\frac{5 - \sqrt{17}}{2} \approx 0.438$ & (\ref{th:dicut_negative}) 		 	& $1/2 + \ee$ 						 & (\ref{th:cardinality_negative}) \\
\hline
Polynomial & \multirow{2}{*}{$1/2 + \ee$} & \multirow{2}{*}{(\ref{th:cardinality_negative})} & \multirow{2}{*}{$1 - 1/e$} & \multirow{2}{*}{\cite{NW78}} & \multirow{2}{*}{$\frac{5 - \sqrt{17}}{2} \approx 0.438$} & \multirow{2}{*}{(\ref{th:dicut_negative})} & \multirow{2}{*}{$0.491$} & \multirow{2}{*}{\cite{GV11}} \\
time hardness &&&&&&&&\\
\hline
\end{tabular}
\caption{Summary of the results for {\CM}. To the right of each result appears the number of the theorem (or references) proving it.} \label{th:cardinality_results}
\end{table*}

\subsection{Related Work\inConference{.}} \label{ssc:related_work}

The literature on submodular maximization problems is very large, and therefore, we mention below only a few of the most relevant works. The classical result of~\cite{NWF78} states that the simple discrete greedy algorithm is a ($1-1/e$)-approximation algorithm for maximizing a monotone submodular function subject to a cardinality constraint. This result is known to be tight~\cite{NW78}, even in the case where the objective function is a coverage function~\cite{F98}. However, when one considers submodular objectives which are not monotone, less is known. An approximation of $0.309$ was given by~\cite{V13}, which was later improved
to $0.325$~\cite{GV11} using a simulated annealing technique. Extending the continuous greedy algorithm of~\cite{CCPV11} to general non-negative submodular objectives, \cite{FNS11} obtained an improved approximation of $1/e-o(1)$. Finally,~\cite{BFNS14} gave a fast $1/e$-approximation algorithm and a much slower $(1/e + 0.004)$-approximation algorithm, demonstrating that $1/e$ is not the right approximation ratio for the problem. On the hardness side, it is known that no polynomial time algorithm can have an approximation ratio better than $0.491$~\cite{GV11}.

For the problem of maximizing a non-negative submodular function subject to no constraints, the first approximation algorithm was presented by Feige et al.~\cite{FMV11} who gave a $2/5$-approximation. This was improved in~\cite{GV11} and~\cite{FNS11b} to $0.41$ and $0.42$, respectively. Finally,~\cite{BFNS12} described a $1/2$-approximation linear time algorithm, matching the hardness result given by~\cite{FMV11}. Huang and Borodin~\cite{HB13} study an online model of this problem where the algorithm can access both $f$ and $\bar{f}$. This model is powerful enough to implement the randomized $1/2$-approximation algorithm of~\cite{BFNS12}, and therefore, \cite{HB13} consider only deterministic algorithms. Azar et al.~\cite{AGR11} consider the Submodular Max-SAT problem, and provide a $2/3$-approximation algorithm for it that can be implemented also in a natural online model. We are not aware of any additional works on online models of submodular maximization.

For the problem of Max-DiCut, Goemans and Williamson~\cite{GW95} obtained $0.796$-approximation using semi-definite programming. This was improved through a series of works~\cite{FG95,LLZ02,MM01} to $0.874$. On the hardness side, a $(12/13 + \ee)$-approximation algorithm will imply $\texttt{P} = \texttt{NP}$~\cite{H01}. Assuming the Unique Games Conjecture, the best possible approximation for Max-Cut is $0.878$~\cite{KKMO07,MOO05}, and this hardness result holds also for Max-DiCut since the last generalizes Max-Cut.

Bar-Noy and Lampis~\cite{BL12} gave a $(1/3)$-competitive deterministic algorithm for an online model of Max-Cut where every revealed node is accompanied by its input and output degrees. For the case of a directed acyclic graph, they provide an improved deterministic algorithm with a competitive ratio of $2/3^{1.5} \approx 0.385$, which is optimal against an adaptive adversary. Huang and Borodin~\cite{HB13} notice that the $(1/3)$-competitive deterministic algorithm of~\cite{BL12} is in fact identical to the $(1/3)$-approximation deterministic algorithm of~\cite{BFNS12} for unconstrained submodular maximization. Using the same ideas, it is not difficult to show that the $(1/2)$-approximation randomized algorithm of~\cite{BFNS12} implies a $(1/2)$-competitive algorithm in this online model. Finally, Feige and Jozeph~\cite{FJ10} consider oblivious algorithms for Max-DiCut---algorithms in which every node is selected into the cut independently with a probability depending solely on its input and output degrees. They show a $0.483$-approximation oblivious algorithm, and prove that no oblivious algorithm has an approximation ratio of $0.4899$.

Finally, the vast literature on buyback problems considers problems that are similar to our model, but assume a linear objective function. Many of these problems are non-trivial only when preemption has a cost, which is usually assumed to be either constant or linear in the value of the preempted element. The work from this literature most closely related to ours is the work of Babaioff et al.~\cite{BHK08,BHK09} who considered a matroid constraint. For other buyback results see, \eg,~\cite{CFMP09,HKM13,IT02,A11,AK09}.

\paragraph{Paper Organization.} Section~\ref{sc:pre} defines additional notation. Section~\ref{sc:unconstrained} gives our results for {\UM} and {\DM}, except for Theorem~\ref{th:unconstrained_no_preemption} whose proof is deferred to Appendix~\ref{sec:unconstrained_no_preemption}. Finally, Section~\ref{sc:uniform} describes our results for {\CM}.
\section{Preliminaries} \label{sc:pre}

We study the following online variant of the submodular maximization problem. There is an unknown ground set of elements $\NN = \{u_1, u_2, \dotsc, u_n\}$, a non-negative submodular function $f : 2^\NN \to \mathbb{R}^+$ and and a down-monotone\footnote{A collection $\II$ of sets is \emph{down-monotone} if $A \subseteq B \subseteq \NN$ and $B \in \II$ imply $A \in \II$.} collection of feasible sets $\II \subseteq 2^\NN$. The objective of the instance is to find a feasible set maximizing $f$. The elements of $\NN$ are revealed one by one. The algorithm creates $n + 1$ feasible solutions: $S_0, S_1, \dotsc, S_n$ (each $S_i\in\II$). The solution $S_0$ is the empty set $\varnothing$. For every $1 \leq i \leq n$, $S_i$ is the solution selected by the algorithm immediately after element $u_i$ is revealed and the algorithm can choose it to be any feasible subset of $S_{i - 1} + u_i$\footnote{Given a set $S$ and an element $u$, we use $S + u $ and $S - u$ as shorthands for $S \cup \{u\}$ and $S \setminus \{u\}$, respectively.}. It is important to note that the algorithm does not know $n$ (the size of the ground set) in advance, hence, the input might end after every element from the algorithm's point of view.
\section{{\SUM} and {\DM}} \label{sc:unconstrained}

Our positive results for {\UM} and {\DM} (\ie, Theorems~\ref{th:unconsrained_positive} and~\ref{th:dicut_positive}) are proved in Section~\ref{ssc:unconstrained_positive}. The negative result for {\DM} (Theorem~\ref{th:dicut_negative}), which applies also to {\UM}, is proved in Section~\ref{ssc:unconstrained_negative}.

\subsection{Algorithms for {\UM} and {\DM}\inConference{.}} \label{ssc:unconstrained_positive}

Before describing our algorithm for {\UM}, we need some notation. For two vectors $x,y \in [0, 1]^\NN$, we use $x \vee y$ and $x \wedge y$ to denote the coordinate-wise maximum and minimum, respectively, of $x$ and $y$ (formally, $(x \vee y)_u = \max\{x_u, y_u\}$ and $(x \wedge y)_u = \min\{x_u, y_u\}$). We abuse notation both in the description of the algorithm and in its analysis, and unify a set with its characteristic vector and an element with the singleton set containing it. The \emph{multilinear} extension of a set function $f : 2^\NN \to \mathbb{R}^+$ is a function $F : [0, 1]^\NN \to \mathbb{R}^+$ defined by $F(x) = \mathbb{E}[f(\RSet(x))]$, where $\RSet(x)$ is a random set containing every element $u \in \NN$ with probability $x_u$, independently. The multilinear extension is an important tool used in many previous works on submodular maximization (see, \eg,~\cite{CCPV11,FNS11,V13}). We denote by $\partial_u F(x)$ the derivative of $F$ at point $x$ with respect to the coordinate corresponding to $u$. It can be checked that $F$ is a multilinear function, and thus:
\[
	\partial_u F(x)
	=
	F(x \vee u) - F(x \wedge (\NN - u))
	\enspace.
\]

Consider Algorithm~\ref{alg:IncreaseFromZero}. Recall that $S_i$ is the solution that the algorithm produces after seeing element $u_i$.

\begin{algorithm}[h!t]
\caption{\textsf{Marginal Choice}} \label{alg:IncreaseFromZero}
\ForEach{element $u_i$ revealed}
{
	Choose a uniformly random threshold $\theta_i \in [0, 1]$.\\
	Let $\NN_i \gets \{u_1, u_2, \dotsc, u_i\}$.\\
	Let $S_i \gets \{u_j \in \NN_i \mid \partial_{u_j} F(\theta_j \cdot \NN_i) \geq 0\}$.
}
\end{algorithm}

Our first objective is to show that Algorithm~\ref{alg:IncreaseFromZero} is an online algorithm \inArxiv{according to}\inConference{of} {\UM}.

\begin{lemma}
For every $1 \leq i \leq n$, $S_i \subseteq S_{i - 1} + u_i$.
\end{lemma}
\begin{proof}
By definition, $S_i$ contains only elements of $\NN_i$. Fix an element $u_j \in \NN_i - u_i$. Then:
\begin{align*}
	u_j \in S_i
	\Rightarrow{} &
	\partial_{u_j} F(\theta_j \cdot \NN_i) \geq 0 \inConference{\\}
	\Rightarrow \inConference{{} &}
	\partial_{u_j} F(\theta_j \cdot \NN_{i - 1}) \geq 0
	\Rightarrow
	u_j \in S_{i - 1}
	\enspace,
\end{align*}
where the second derivation follows from submodularity.
\end{proof}

Next, we bound the competitive ratio of Algorithm~\ref{alg:IncreaseFromZero}. Fix an element $u_i \in \NN$. By submodularity, $g(z) = \partial_{u_i} F(z \cdot \NN)$ is a continuous non-increasing function of $z$. Hence, by the intermediate value theorem, one of the following must hold: $g(z)$ is always positive in the range $[0, 1]$, $g(z)$ is always negative in the range $[0, 1]$ or $g(z)$ has at least one root $z_0 \in [0, 1]$. In the last case, the set $I_0 \subseteq [0, 1]$ of the roots of $g(z)$ is non-empty. Moreover, by the monotonicity and continuity of $g(z)$, $I_0$ is a closed range. Using these observations, we define a vector $y \in [0, 1]^\NN$ as follows:
\[
	y_{u_i}
	=
	\begin{cases}
		0 & \text{if $\partial_{u_i} F(0 \cdot \NN) < 0$} \enspace, \\
		1 & \text{if $\partial_{u_i} F(1 \cdot \NN) > 0$} \enspace, \\
		\max I_0 & \text{otherwise} \enspace.
	\end{cases}
\]

\noindent \textbf{Remark:} Notice that the case $\partial_{u_i} F(0 \cdot \NN) < 0$ in the definition of $y$ can happen only when $f(\varnothing) > 0$.

\begin{observation} \label{obs:y_alg_equivalent}
Every element $u_i \in \NN$ belongs to $S_n$ with probability $y_{u_i}$, independently. Hence, $\mathbb{E}[f(S_n)] = F(y)$.
\end{observation}
\begin{proof}
An element $u_i \in \NN$ belongs to $S_n$ if and only if $\partial_{u_i} F(\theta_i \cdot \NN) \geq 0$, which is equivalent to the condition $\theta_i \leq y_{u_i}$. Clearly, the last condition happens with probability $y_{u_i}$, and is independent for different elements.
\end{proof}

The last observation implies that analyzing Algorithm~\ref{alg:IncreaseFromZero} is equivalent to lower bounding $F(y)$.

\begin{lemma} \label{le:all_worse_than_algorithm}
For every $\lambda \in [0, 1]$, $F(y \wedge (\lambda \cdot \NN)) \geq F(\lambda \cdot \NN)$.
\end{lemma}
\begin{proof}
Observe that:
\begin{align*}
	F(y \wedge (\lambda \cdot \NN))
	={} &
	f(\varnothing) + \int_0^\lambda \frac{d F(y \wedge (z \cdot \NN))}{d z} dz \inConference{\\}
	=\inConference{{} &}
	f(\varnothing) + \int_0^\lambda \sum_{\substack{u \in \NN \\ z \leq y_u}} \partial_u F(y \wedge (z \cdot \NN)) dz
	\enspace,
\end{align*}
where the second equality is due to the chain rule. By submoduliry and the observation that $\partial_u F(z \cdot \NN)$ is non-positive for every $z > y_u$, we get:
\begin{align*}
	\sum_{\substack{u \in \NN \\ z \leq y_u}} \partial_u F(y \wedge (z \cdot \NN))
	\geq{} &
	\sum_{\substack{u \in \NN \\ z \leq y_u}} \partial_u F(z \cdot \NN)\inConference{\\}
	\geq\inConference{{} &}
	\sum_{u \in \NN} \partial_u F(z \cdot \NN)
	\enspace.
\end{align*}
Combining the above equality and inequality, and using the chain rule again, gives:
\begin{align*}
	F(y \wedge (\lambda \cdot \NN))\inConference{&}
	\geq\inArxiv{{} &}
	f(\varnothing) + \int_0^\lambda \left[\sum_{u \in \NN} \partial_u F(z \cdot \NN)\right] dz\inConference{\\}
	=\inConference{{} &}
	f(\varnothing) + \int_0^\lambda \frac{d F(z \cdot \NN)}{d z} dz
	=
	F(\lambda \cdot \NN)
	\enspace.
	\qedhere
\end{align*}
\end{proof}

We also need a lemma proved by~\cite{FNS11}.

\begin{lemma}[Lemma~III.5 of~\cite{FNS11}] \label{le:y_max_bound}
Consider a vector $x \in [0, 1]^\NN$. Assuming $x_u \leq a$ for every $u \in \NN$, then for every set $S \subseteq \NN$, $F(x \vee S) \geq (1 - a)f(S)$.
\end{lemma}

The following corollary follows from the last two lemmata.

\begin{corollary} \label{cor:unconstrained_ratio}
$F(y) \geq e^{-1} \cdot F(OPT)$, where $OPT$ is the subset of $\NN$ maximizing $f(OPT)$.
\end{corollary}
\begin{proof}
Fix two values $0 \leq a \leq b \leq 1$. By the chain rule:
\begin{align} \label{eq:integral_form}
	\inConference{&}F(y \wedge (b \cdot \NN)) -  F(y \wedge (a \cdot \NN))\inConference{\\}
	={} &
	\int_a^b \frac{d F(y \wedge (z \cdot \NN))}{dz} dz \inConference{\nonumber} \\
	={} &
	\int_a^b \sum_{\substack{u \in \NN \\ z \leq y_u}} \partial_u F(y \wedge (z \cdot \NN)) dz\inConference{\nonumber \\}
	\geq\inConference{{} &}
	\int_a^b \sum_{\substack{u \in \NN \\ z \leq y_u}} \partial_u F(z \cdot \NN) dz \nonumber
	\enspace,
\end{align}
where the inequality follows from submodularity. We use two ways to lower bound the rightmost hand side of the above inequality. First, observe that $\partial_u F(z \cdot \NN) \geq 0$ whenever $z < y_u$. Thus, $F(y \wedge (z \cdot \NN))$ is a non-decreasing function of $z$. Additionally, $\partial_u F(z \cdot \NN) \leq 0$ whenever $z > y_u$ and $\partial_u F(z \cdot \NN) \geq 0$ whenever $z < y_u$. This allows us to derive the second lower bound, which holds whenever $z$ is not equal to any coordinate of $y$ (\ie, for every value of $z$ except for a finite set of values):
\begin{align*}
	\sum_{\substack{u \in \NN \\ z \leq y_u}} \partial_u F(z \cdot \NN)\inConference{&}
	\geq\inArxiv{{} &}
	\sum_{u \in OPT} \partial_u F(z \cdot \NN)\inConference{\\}
	=\inConference{{} &}
	\frac{\sum_{u \in OPT} [F(u \vee (z \cdot \NN)) - F(z \cdot \NN)]}{1 - z}\\
	\geq{} &
	\frac{F(OPT \vee (z \cdot \NN)) - F(z \cdot \NN)}{1 - z}\inConference{\\}
	\geq\inConference{{} &}
	f(OPT) - \frac{F(y \wedge (z \cdot \NN))}{1 - z}
	\enspace,
\end{align*}
where the last inequality follows from Lemmata~\ref{le:all_worse_than_algorithm} and~\ref{le:y_max_bound}. Plugging this lower bound and $a = 0$ into~\eqref{eq:integral_form} gives:
\begin{align*}
	&
	F(y \wedge (b \cdot \NN)) \inConference{\\}
	\geq\inConference{{} &}
	f(\varnothing) + \int_0^b \left[f(OPT) - \frac{F(y \wedge (z \cdot \NN))}{1 - z}\right] dz
	\enspace.
\end{align*}
The solution of this differential equation is:
\[
	F(y \wedge (b \cdot \NN))
	\geq
	-(1 - b) \cdot \ln (1 - b) \cdot f(OPT)
	\enspace.
\]
The rightmost hand side of the last inequality is maximized for $b = 1 - e^{-1}$. Recalling that $F(y \wedge (z \cdot \NN))$ is a non-decreasing function of $z$, we get:
\begin{align*}
	F(y)
	={} &
	F(y \wedge (1 \cdot \NN))\inConference{\\}
	\geq\inConference{{} &}
	F(y \wedge ((1 - e^{-1}) \cdot \NN))
	\geq
	e^{-1} \cdot f(OPT)
	\enspace.
	\qedhere
\end{align*}
\end{proof}

The last corollary completes the proof of the first part of Theorem~\ref{th:unconsrained_positive}. A na\"{\i}ve implementation of Algorithm~\ref{alg:IncreaseFromZero} requires an exponential time to evaluate $F$. However, for every constant $c > 0$ it is possible to approximate, in polynomial time, $\partial_{u_j} F(\theta_j \cdot \NN_i)$ with probability $1 - c \ee i^{-2}$ up to an additive error of $c\ee i^{-2} \cdot f(OPT)$ using sampling (see~\cite{CCPV11} for an example of a similar estimate).\footnote{It is not possible to decrease the error to something that depends on $n$ because $n$ is unknown to the algorithm when it calculates these estimations.} For a small enough $c$, such an approximation is good enough to affect the competitive ratio of the algorithm by only an $\ee$. To keep the algorithm within {\UM}, the samples used to approximate $\partial_{u_j} F(\theta_j \cdot \NN_i)$ for different values of $i$ need to be correlated to guarantee that the approximation is a decreasing function of $i$. \inConference{We omit details from this extended abstract.}\inArxiv{For more details, see Appendix~\ref{app:unconstrained_polynomial}.}

%
%

The rest of this section considers {\DM} and is devoted for proving Theorem~\ref{th:dicut_positive}. Recall that in this model $f$ is the cut function of some weighted directed graph $G = (V, A)$. Let $\Win(u)$ ($\Wout(u)$) denote the total weight of the arcs entering (leaving) element $u$. Throughout this section we assume there are no arcs of $G$ leaving nodes of $V \setminus \NN$. Removing such arcs does not affect the value of any solution, and thus, our results hold also without the assumption. Feige and Jozeph~\cite{FJ13} proved the following theorem (we rephrased it in our notation).
\begin{theorem} \label{th:oblivious}
There exists a non-decreasing function $h : [0, 1] \to [0, 1]$ such that the vector $y \in [0, 1]^\NN$ defined by:
\[
	y_u
	=
	h\left( \frac{\Wout(u)}{\Win(u) + \Wout(u)} \right)
\]
obeys $F(y) \geq 0.483 \cdot f(OPT)$. Moreover, $h$ is independent of $G$ and can be computed in constant time (assuming numbers comparison takes constant time).
\end{theorem}

Notice that in the last theorem $y_u$ is undefined when $\Win(u) = \Wout(u) = 0$. The theorem holds regardless of the value used for $y_u$ in this case because such an element $u$ is isolated. We can now present our algorithm for {\DM} which is depicted as Algorithm~\ref{alg:DegreesRatio}.

\begin{algorithm}[h!t]
\caption{\textsf{Degrees Ratio Choice}} \label{alg:DegreesRatio}
\DontPrintSemicolon
\ForEach{element $u_i$ revealed}
{
	Choose a uniformly random threshold $\theta_i \in [0, 1]$.\\
	Initialize $S_i \leftarrow \varnothing$.\\
	\For{$j$ = $1$ \KwTo $i$}
	{
		Let $\Win(u_j, i)$ denote the total weight of arcs from \emph{revealed} elements to $u_j$.\\
		\lIf{$\Win(u_j, i) + \Wout(u_j) > 0$}
		{
			Let $y_{u_j}(i) \gets h\left( \frac{\Wout(u_j)}{\Win(u_j, i) + \Wout(u_j)} \right)$.
		}
		\lElse
		{
			Let $y_{u_j}(i) \gets 1$.
		}
		
		\BlankLine
		
		\lIf{$y_{u_j}(i) \geq \theta_j$}
		{
			Add $u_j$ to $S_i$.
		}
	}
}
\end{algorithm}

\begin{observation}
The competitive ratio of Algorithm~\ref{alg:DegreesRatio} is at least $0.483$.
\end{observation}
\begin{proof}
Notice that the vector $y(n)$ defined by the algorithm is equal to the vector $y$ defined by Theorem~\ref{th:oblivious} (wherever the last is defined). Notice also that every element $u \in \NN$ belongs to $S_n$ with probability $y_u(n)$, independently. Hence, $\mathbb{E}[f(S_n)] = F(y)$. The observation now follows from Theorem~\ref{th:oblivious}.
\end{proof}

To complete the proof of Theorem~\ref{th:dicut_positive}, we only need the following lemma which shows that Algorithm~\ref{alg:DegreesRatio} is an online algorithm of {\DM}.

\begin{lemma}
For every $1 \leq i \leq n$, $S_i \subseteq S_{i - 1} + u_i$.
\end{lemma}
\begin{proof}
For $i = 1$ there is nothing to prove since $S_1 \subseteq \{u_1\}$. Thus, we assume from now on $i \geq 2$. By definition, $S_i$ contains only elements of $\NN_i$. Fix an element $u_j \in \NN_i - u_i$. We have to show that $u_j \in S_{i - 1}$. First, let us prove that:
\begin{equation} \label{eq:y_decrease}
	y_{u_j}(i) \leq y_{u_j}(i - 1)
	\enspace.
\end{equation}
There are two cases to consider. If $\Win(u_j, i - 1) + \Wout(u_j) \leq 0$ then $y_{u_j}(i - 1) = 1$, which proves~\eqref{eq:y_decrease} since $y_{u_j}(i) \in [0, 1]$. Otherwise, since $\Win(u_j, i)$ is a non-decreasing function of $i$ we get $\Win(u_j, i) + \Wout(u_j) > 0$. Inequality~\eqref{eq:y_decrease} now follows since $h$ is also non-decreasing. Using~\eqref{eq:y_decrease} we now get:
\[
	u_j \in S_i
	\Rightarrow
	y_{u_j}(i) \geq \theta_j
	\Rightarrow
	y_{u_j}(i - 1) \geq \theta_j
	\Rightarrow
	u_j \in S_{i - 1}
	\enspace.
	\qedhere
\]
\end{proof}

\subsection{Hardness results for {\DM}\inConference{.}} \label{ssc:unconstrained_negative} 

In this section we prove Theorem~\ref{th:dicut_negative}. The proof of the theorem is split between two lemmata. Lemma~\ref{le:digraph_negative_deterministic} proves the part of Theorem~\ref{th:dicut_negative} referring to deterministic algorithms, and Lemma~\ref{le:digraph_negative_randomized} proves the part referring to randomized algorithms. Both lemmata present an absolute graph $G = (V, A)$, and then fix an algorithm $ALG$ and describe an adversary that reveals some of the nodes of $V$. The choice which nodes to reveal is done iteratively, \ie, the adversary initially reveals some nodes and then can reveal additional nodes based on the decisions of $ALG$ (or the probabilities of various decisions in case $ALG$ is randomized). Formally, in the hard instance for $ALG$, the set $\NN$ is the set of nodes revealed by the adversary, and the hard instance reveals these nodes in the same order they are revealed by the adversary.

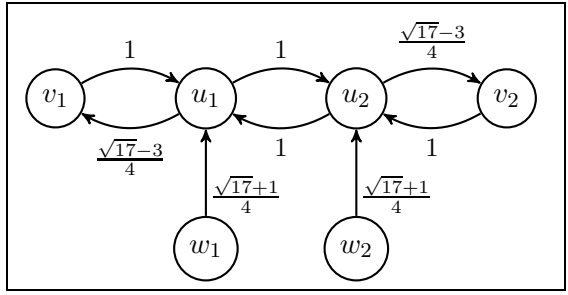
\begin{figure} 
\center
\fbox{
	\begin{tikzpicture}[->,>=stealth',thick]

	\draw
		node[node] (v1) {$v_1$}
		node[node] (u1) [right of=v1] {$u_1$}
		node[node] (u2) [right of=u1] {$u_2$}
		node[node] (v2) [right of=u2] {$v_2$}
		node[node] (w1) [below of=u1] {$w_1$}
		node[node] (w2) [below of=u2] {$w_2$}
	;

  \path[every node/.style={font=\sffamily\small}]
    (v1) edge[bend left] node[above] {$1$} (u1)
		(u1) edge[bend left] node[below] {$\frac{\sqrt{17} - 3}{4}$} (v1)
    (u1) edge[bend left] node[above] {$1$} (u2)
    (u2) edge[bend left] node[below] {$1$} (u1)
		(u2) edge[bend left] node[above] {$\frac{\sqrt{17} - 3}{4}$} (v2)
    (v2) edge[bend left] node[below] {$1$} (u2)
		(w1) edge node[right, pos=0.25] {\hspace{-0.2cm} $\frac{\sqrt{17} + 1}{4}$} (u1)
		(w2) edge node[right, pos=0.25] {\hspace{-0.2cm} $\frac{\sqrt{17} + 1}{4}$} (u2)
	;
\end{tikzpicture}
}
\caption{Graph for Lemma~\ref{le:digraph_negative_deterministic}. The number beside each arc represents its weight.} \label{fig:digraph_negative_deterministic}
\end{figure}

\begin{lemma} \label{le:digraph_negative_deterministic}
\inArxiv{No}\inConference{The competitive ratio of no} deterministic online algorithm \inArxiv{has a competitive ratio}\inConference{is} better than $(5 - \sqrt{17})/2$ for {\DM}.
\end{lemma}
\begin{proof}
Consider the directed graph $G$ describe in Figure~\ref{fig:digraph_negative_deterministic}, and fix a deterministic algorithm $ALG$ for {\DM}. We describe an adversary that reveals some of the nodes of $G$ and forces $ALG$ to be no better than $((5 - \sqrt{17}) / 2)$-competitive. For every pair of nodes $x$ and $y$ of $G$, we denote by $c(xy)$ the weight of the arc from $x$ to $y$. Our adversary begins by revealing $u_1$ and $u_2$ and stops if $ALG$'s solution at this point is either empty or equal to $\{u_1, u_2\}$. If $ALG$'s solution is empty then $ALG$ is not competitive at all. On the other hand, if $ALG$'s solution is $\{u_1, u_2\}$ then its solution has a value of:
\[
	c(u_1v_1) + c(u_2v_2)
	=
	2 \cdot \frac{\sqrt{17} - 3}{4}
	=
	\frac{\sqrt{17} - 3}{2}
	\enspace.
\]
On the other hand, the optimal solution is $\{u_1\}$, whose value is:
\[
	c(u_1v_1) + c(u_1u_2)
	=
	\frac{\sqrt{17} - 3}{4} + 1
	=
	\frac{\sqrt{17} + 1}{4}
	\enspace.
\]
Thus, in this case, the competitive ratio of $ALG$ is at most:
\[
	\frac{\sqrt{17} - 3}{2} : \frac{\sqrt{17} + 1}{4}
	=
	\frac{2\sqrt{17} - 6}{\sqrt{17} + 1}
	=
	\frac{5 - \sqrt{17}}{2}
	\enspace.
\]

By the above discussion, the only interesting case is when $ALG$'s solution contains exactly one of the elements $u_1$ or $u_2$. Notice that $G$ is symmetric in the sense that switching every node having the index $1$ with the corresponding node having the weight $2$ does not change the graph. Hence, it is safe to assume $ALG$'s solution is exactly $\{u_1\}$. The next step of the adversary is to reveal $v_1$. If $v_1$ enters the solution of $ALG$, then, regardless of weather $u_1$ is kept in the solution, the value of $ALG$'s solution is $c(u_1u_2) = c(v_1u_1) = 1$. On the other hand, the optimal solution at this point is $\{u_2, v_1\}$ whose value is:
\[
	c(u_2u_1) + c(u_2v_2) + c(v_1u_1)
	=
	1 + \frac{\sqrt{17} - 3}{4} + 1
	=
	\frac{\sqrt{17} + 5}{4}
	\enspace.
\]
Hence, the competitive ratio of $ALG$ is at most:
\[
	\frac{4}{\sqrt{17} + 5}
	=
	\frac{5 - \sqrt{17}}{2}
	\enspace.
\]

We are left to handle the case in which $v_1$ does not enter $ALG$'s solution. In this case, the adversary reveals also $w_1$. $ALG$'s solution at this point must be a subset of $\{u_1, w_1\}$ and every such subset has a value of at most:
\[
	c(u_1v_1) + c(u_1u_2) = c(w_1u_1)
	=
	\frac{\sqrt{17} + 1}{4}
	\enspace.
\]
On the other hand, the optimal solution at this point is $\{u_2, v_1, w_1\}$ whose value is:
\begin{align*}
	&
	c(u_2u_1) + c(u_2v_2) + c(v_1u_1) + c(w_1u_1)\inConference{\\}
	=\inConference{{} &}
	1 + \frac{\sqrt{17} - 3}{4} + 1 + \frac{\sqrt{17} + 1}{4}
	=
	\frac{\sqrt{17} + 3}{2}
	\enspace.
\end{align*}
Hence, the competitive ratio of $ALG$ is at most:
\begin{align*}
	\frac{\sqrt{17} + 1}{4} : \frac{\sqrt{17} + 3}{2}
	={} &
	\frac{\sqrt{17} + 1}{2\sqrt{17} + 6}\inConference{\\}
	=\inConference{{} &}
	\frac{7 - \sqrt{17}}{8}
	<
	\frac{5 - \sqrt{17}}{2}
	\enspace.
	\qedhere
\end{align*}
\end{proof}

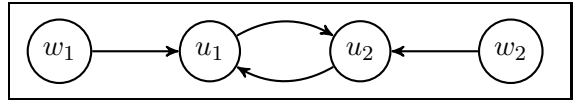
\begin{figure} 
\center
\fbox{
	\begin{tikzpicture}[->,>=stealth',thick]

	\draw
		node[node] (w1) {$w_1$}
		node[node] (u1) [right of=w1] {$u_1$}
		node[node] (u2) [right of=u1] {$u_2$}
		node[node] (w2) [right of=u2] {$w_2$}
	;

  \path
    (w1) edge node {} (u1)
    (u1) edge[bend left] node {} (u2)
    (u2) edge[bend left] node {} (u1)
    (w2) edge node {} (u2)
	;
\end{tikzpicture}
}
\caption{Graph for Lemma~\ref{le:digraph_negative_randomized}. All arcs have an identical weight assumed to be $1$.} \label{fig:digraph_negative_randomized}
\end{figure}

\begin{lemma} \label{le:digraph_negative_randomized}
No randomized online algorithm has a competitive ratio better than $4/5$ for {\DM}.
\end{lemma}
\begin{proof}
Consider the directed graph $G$ describe in Figure~\ref{fig:digraph_negative_randomized}, and fix an algorithm $ALG$ for {\DM}. We describe an adversary that reveals some of the nodes of $G$ and forces $ALG$ to be no better than $4/5$-competitive. The adversary begins by revealing the nodes $u_1$ and $u_2$ of the graph. At this point, let $p_1$ be the probability that $u_1$ is alone in the solution of $ALG$ and $p_2$ be the probability that $u_2$ is alone in this solution. If $p_1 + p_2 \leq 4/5$ then the adversary stops here, and the competitive ratio of $ALG$ is no better than $4/5$ since the optimal solution at this point has a value of $1$ and the expected value of the $ALG$'s solution is only $p_1 + p_2$. Thus, we may assume from now on $p_1 + p_2 > 4/5$.

By symmetry, we may also assume, without loss of generality, $p_2 \leq p_1$, which implies: $2p_1 \geq p_1 + p_2 > 4/5$. The next step of the adversary is revealing $w_1$. The optimal solution at this point is $\{u_2, w_1\}$, whose value is $2$. Let analyze the best solution $ALG$ can produce. With probability $p_1$, $ALG$ must pick a subset of $\{u_1, w_1\}$, and thus, cannot end up with a solution of value better than $1$. With the remaining probability, $ALG$ produces a solution no better than the optimal solution, \ie, its value is at most $2$. Hence, the expected value of $ALG$'s solution is at most:
\[
	p_1 + 2(1 - p_1)
	=
	2 - p_1
	<
	2 - 2/5
	=
	2 \cdot (4/5)
	\enspace.
	\qedhere
\]
\end{proof}
\section{\SCM} \label{sc:uniform}

Our positive and negative results for {\CM} (\ie, Theorems~\ref{th:cardinality_positive} and~\ref{th:cardinality_negative}, respectively) are proved in Sections~\ref{ssc:cardinality_positive} and~\ref{ssc:cardinality_negative}, respectively.

\subsection{Algorithms for {\CM}\inConference{.}} \label{ssc:cardinality_positive}

We first consider the special case of {\CM} where the objective function $f$ is monotone (in addition to being non-negative and submodular). For this case it is possible to derive a $1/4$-competitive algorithm from the works of Ashwinkumar~\cite{A11} and Chakrabarti and Kale~\cite{CK13} on streaming algorithms. We describe an alternative algorithm which achieves the same competitive ratio and has an easier analysis. A variant of this algorithm is applied later in this section to {\GCM}.

Our alternative algorithm has a parameter $c > 0$ and is given as Algorithm~\ref{alg:GreedyWithThreshold}. Intuitively the algorithm gets the first $k$ revealed elements unconditionally. For every element $u_i$ arriving after the first $k$ elements, the algorithm finds the best swap replacing an element of the current solution with $u_i$. If this swap increases the value of the current solution by enough to pass a given threshold, then it is carried out and $u_i$ enters the solution. 

\begin{algorithm*}[h!t]
\caption{\textsf{Greedy with Theshold}$(c)$} \label{alg:GreedyWithThreshold}
Let $S_0 \gets \varnothing$.\\
\ForEach{element $u_i$ revealed}
{
	\If{$i \leq k$}
	{
		Let $S_i \gets S_{i - 1} + u_i$.\\
	}
	\Else
	{
		Let $u'_i$ be the element of $S_{i - 1}$ maximizing $f(S_{i - 1} + u_i - u'_i)$.\\
		\If{$f(S_{i - 1} + u_i - u'_i) - f(S_{i - 1}) \geq c \cdot f(S_{i - 1}) / k$}
		{
			Let $S_i \gets S_{i - 1} + u_i - u'_i$.
		}
		\Else
		{
			Let $S_i \gets S_{i - 1}$.
		}
	}
}
\end{algorithm*}

Let us define some notation. For every $0 \leq i \leq n$, let $A_i = \cup_{j = 1}^i S_j$. Informally, $A_i$ is the set of all elements of $\{u_1, u_2, \dotsc, u_i\}$ originally accepted by Algorithm~\ref{alg:GreedyWithThreshold}, regardless of whether they are preempted at a later stage or not. We also use $f(u \mid S)$ to denote the marginal contribution of an element $u$ to a set $S$. Formally, $f(u \mid S) = f(S + u) - f(S)$. The following technical lemma is used later in the proof.

\begin{lemma} \label{le:maximum_better_than_average}
For every $k < i \leq n$, $f(S_{i - 1} + u_i - u'_i) - f(S_{i - 1}) \geq f(u_i \mid A_{i - 1}) - f(S_{i - 1}) / k$.
\end{lemma}
\begin{proof}
By the choice of $u'_i$ and an averaging argument,
\begin{align*}
	\inConference{&}
	f(S_{i - 1} + u_i - u'_i) \inArxiv{&}- f(S_{i - 1}) \inConference{\\}
	\geq\inConference{{} &}
	\frac{\sum_{u' \in S_{i - 1}} [f(S_{i - 1} + u_i - u') - f(S_{i - 1})]}{k}\\
	\geq{} &
	\frac{\sum_{u' \in S_{i - 1}} [f(S_{i - 1} + u_i) - f(S_{i - 1})]}{k} \inConference{\\ &} + \frac{\sum_{u' \in S_{i - 1}} [f(S_{i - 1} - u') - f(S_{i - 1})]}{k}
	\enspace,
\end{align*}
where the second inequality follows from submodularity. The first term on the rightmost hand side is $f(u_i \mid S_{i - 1})$, which can be lower bounded by $f(u_i \mid A_{i - 1})$  because $S_{i - 1} \subseteq A_{i - 1}$. The second term on the rightmost hand side can be lower bounded using the submodularity and non-negativity of $f$ as follows:
\begin{align*}
	\frac{\sum_{u' \in S_{i - 1}} [f(S_{i - 1} - u') - f(S_{i - 1})]}{k}
	\geq{} &
	\frac{f(\varnothing) - f(S_{i - 1})}{k}\inConference{\\}
	\geq\inConference{{} &}
	- \frac{f(S_{i - 1})}{k}
	\enspace.
	\qedhere
\end{align*}
\end{proof}

The following lemma relates $S_i$ and $A_i$.

\begin{lemma} \label{le:S_fraction_of_A}
For every $0 \leq i \leq n$, $f(S_i) \geq \frac{c}{c + 1} \cdot f(A_i)$.
\end{lemma}
\begin{proof}
For $i = 0$, the lemma clearly holds since $f(A_0) = f(\varnothing) = f(S_0)$. Thus, it is enough to prove that for every $1 \leq i \leq n$, $f(S_i) - f(S_{i - 1}) \geq \frac{c}{c + 1} \cdot [f(A_i) - f(A_{i - 1})]$. If the algorithm does not accept $u_i$ then $S_i = S_{i - 1}$ and  $A_i = A_{i - 1}$, and the claim is trivial. Hence, we only need to consider the case in which the algorithm accepts element $u_i$.

If $i \leq k$ then $f(S_i) - f(S_{i - 1}) = f(u_i \mid S_{i - 1}) \geq f(u_i \mid A_{i - 1}) = f(A_i) - f(A_{i - 1})$, where the inequality follows from submodularity since $S_{i - 1} \subseteq A_{i - 1}$. If $i > k$, then it is possible to lower bound $f(S_i) - f(S_{i - 1}) = f(S_{i - 1} + u_i - u'_i) - f(S_{i - 1})$ in two ways. First, a lower bound of $f(u_i \mid A_{i - 1}) - f(S_{i - 1}) / k$ is given by Lemma~\ref{le:maximum_better_than_average}. A second lower bound of $c \cdot f(S_{i - 1}) / k$ follows since the algorithm accepts $u_i$. Using both lower bounds, we get:
\begin{align*}
	\inConference{&}
	f(S_i) - f(S_{i - 1})\inConference{\\}
	\geq{} &
	\max\{f(u_i \mid A_{i - 1}) - f(S_{i - 1})/k, c \cdot f(S_{i - 1})/k\}\\
	\geq{} &
	\frac{c \cdot [f(u_i \mid A_{i - 1}) - f(S_{i - 1})/k] + 1 \cdot [c \cdot f(S_{i - 1})/k]}{c + 1}\\
	\geq{} &
	\frac{c}{c + 1} \cdot f(u_i \mid A_{i - 1})
	=
	\frac{c}{c + 1} \cdot [f(A_i) - f(A_{i - 1})]
	\enspace.
	\qedhere
\end{align*}
\end{proof}

We are now ready to prove the competitive ratio of Algorithm~\ref{alg:GreedyWithThreshold}, and thus, also the second part of Theorem~\ref{th:cardinality_positive}.

\begin{corollary}
The competitive ratio of Algorithm~\ref{alg:GreedyWithThreshold} is at least $\frac{c}{(c + 1)^2}$. Hence, for $c = 1$ the competitive ratio of Algorithm~\ref{alg:GreedyWithThreshold} is at least $1/4$.
\end{corollary}
\begin{proof}
Let $OPT$ be the optimal solution, and consider an element $u_i \in OPT \setminus A_i$. Since $u_i$ was rejected by Algorithm~\ref{alg:GreedyWithThreshold}, the following inequality must hold:
\begin{align*}
	\frac{c \cdot f(S_{i - 1})}{k}
	>{} &
	f(S_{i - 1} + u_i - u'_i) - f(S_{i - 1}) \inConference{\\}
	\geq\inConference{{} &}
	f(u_i \mid A_{i - 1}) - \frac{f(S_{i - 1})}{k}
	\enspace,
\end{align*}
where the second inequality follows from Lemma~\ref{le:maximum_better_than_average}. Rearranging yields:
\[
	f(u_i \mid A_{i - 1})
	<
	\frac{c+1}{k} \cdot f(S_{i - 1})
	\leq
	\frac{c+1}{k} \cdot f(S_n)
	\enspace,
\]
where the last inequality uses the monotonicity of $f(S_i)$ (as a function of $i$). In conclusion, by the submodularity and monotonicity of $f$ and Lemma~\ref{le:S_fraction_of_A},
\begin{align*}
	f(OPT)
	\leq{} &
	f(A_n) + \sum_{u \in OPT \setminus A_n} f(u \mid A_n) \inConference{\\}
	<\inConference{{} &}
	f(A_n) + \sum_{u \in OPT \setminus A_n} \left(\frac{c + 1}{k} \cdot f(S_n)\right)\\
	\leq{} &
	f(A_n) + (c + 1) \cdot f(S_n)
	\leq
	\frac{(c + 1)^2}{c} \cdot f(S_n)
	\enspace.
	\qedhere
\end{align*}
\end{proof}

At this point we return the focus to {\GCM} allowing the objective $f$ to be non-monotone. For the algorithm and analysis we present we need the notation described at the beginning of Section~\ref{sc:unconstrained}. More specifically, we use the multilinear extension $F$ of $f$ and the assumption that a set $S \subseteq \NN$ represents also its characteristic vector. As a first step, consider a variant of Algorithm~\ref{alg:GreedyWithThreshold} modified in two ways:
\begin{compactitem}
	\item The value $k$ is replaced with $pk$ for an integer parameter $p \geq 1$.
	\item The algorithm uses an auxiliary function $g : 2^\NN \to \mathbb{R}^+$ instead of the real objective function $f$. The auxiliary function $g$ is defined as $g = F(p^{-1} \cdot S)$.
\end{compactitem}

Notice that the modified algorithm does not produce a feasible solution. Still, we are interested in analyzing it for reasons that are explained later. The proofs of Lemmata~\ref{le:maximum_better_than_average} and~\ref{le:S_fraction_of_A} do not use the monotonicity of $f$. Hence, both lemmata still hold with $pk$ and $g$ taking the roles of $k$ and $f$, respectively (notice that $g$ is non-negative and submodular). The resulting lemmata after the replacements are given below.

\begin{lemma} \label{le:maximum_better_than_average_random}
For every $pk < i \leq n$, $g(S_{i - 1} + u_i - u'_i) - g(S_{i - 1}) \geq g(u_i \mid A_{i - 1}) - g(S_{i - 1}) / (pk)$.
\end{lemma}

\begin{lemma} \label{le:S_fraction_of_A_random}
For every $0 \leq i \leq n$, $g(S_i) \geq \frac{c}{c + 1} \cdot g(A_i)$.
\end{lemma}

We also need the following lemma of~\cite{FMV11}.
\begin{lemma}[Lemma~2.3 of~\cite{FMV11}] \label{le:sampling_lowerbound}
Let $f : 2^\NN \to \mathbb{R}$ be submodular, $A,B \subseteq \NN$ two (not necessarily disjoint) sets and $A(p),B(q)$ their independently sampled subsets, where each element of $A$ appears in $A(p)$ with probability $p$ and each element of $B$ appears in $B(q)$ with probability $q$. Then
\begin{align*}
	&
	\mathbb{E}[f(A(p) \cup B(q))]
	\geq 
	(1 - p)(1 - q) \cdot f(\varnothing) \inConference{\\ &} + p(1 - q) \cdot f(A) + (1 - p)q \cdot f(B) + pq \cdot f(A \cup B)
	\enspace.
\end{align*}
\end{lemma}

Using the above lemmata, it is possible to get a guarantee on $g(S_n)$.

\begin{corollary} \label{co:S_n_guarantee}
Let $OPT$ be the optimal solution. The modified Algorithm~\ref{alg:GreedyWithThreshold} guarantees:
\[
	g(S_n) \geq \frac{cp^{-1}(1 - p^{-1})}{(c + 1)(1 + p^{-1}c)} \cdot f(OPT)
	\enspace.
\]
Hence, for $c = 7/4$ and $p = 3$, $g(S_n) \geq (56/627) \cdot f(OPT)$.
\end{corollary}
\begin{proof}
Consider an element $u_i \in OPT \setminus A_i$. Since $u_i$ was rejected by Algorithm~\ref{alg:GreedyWithThreshold}, the following inequality must hold:
\begin{align*}
	\frac{c \cdot g(S_{i - 1})}{pk}
	>{} &
	g(S_{i - 1} + u_i - u'_i) - g(S_{i - 1})\inConference{\\}
	\geq\inConference{{} &}
	g(u_i \mid A_{i - 1}) - \frac{g(S_{i - 1})}{pk}
	\enspace,
\end{align*}
where the second inequality follows from Lemma~\ref{le:maximum_better_than_average_random}. Rearranging yields:
\[
	g(u_i \mid A_{i - 1})
	<
	\frac{c+1}{pk} \cdot g(S_{i - 1})
	\leq
	\frac{c+1}{pk} \cdot g(S_n)
	\enspace,
\]
where the last inequality uses the monotonicity of $g(S_i)$ (as a function of $i$) for $i \geq pk$. By the submodularity of $g$ and Lemma~\ref{le:S_fraction_of_A_random}, we now get:
\begin{align*}
	g(A_n \cup OPT)
	\leq{} &
	g(A_n) + \sum_{u \in OPT \setminus A_n} g(u_i \mid A_n)\inConference{\\}
	<\inConference{{} &}
	g(A_n) + \sum_{u \in OPT \setminus A_n} \left(\frac{c + 1}{pk} \cdot g(S_n)\right)\\
	\leq{} &
	g(A_n) + p^{-1}(c + 1) \cdot g(S_n)\inConference{\\}
	\leq\inConference{{} &}
	\frac{(c + 1)(1 + p^{-1}c)}{c} \cdot g(S_n)
	\enspace.
\end{align*}

On the other hand, by Lemma~\ref{le:sampling_lowerbound},
\begin{align*}
	g(A_n \cup OPT)
	={} &
	F(p^{-1} \cdot OPT + p^{-1} \cdot (A_n \setminus OPT))\inConference{\\}
	\geq\inConference{{} &}
	p^{-1}(1 - p^{-1}) \cdot f(OPT)
	\enspace.
	\qedhere
\end{align*}
\end{proof}

A na\"{\i}ve implementation of Algorithm~\ref{alg:GreedyWithThreshold} requires an exponential time to evaluate $F$. However, this can be fixed, at a loss of an arbitrary small constant $\ee > 0$ in the guarantee of Corollary~\ref{co:S_n_guarantee}. See the above discussion about a polynomial time implementation of Algorithm~\ref{alg:IncreaseFromZero} for the main ideas necessary for such an implementation.\inArxiv{ For more details, see Appendix~\ref{app:uniform_polynomial}.}

Consider a collection of vectors $\{y_i\}_{i = 0}^n$ obtained from the sets $\{S_i\}_{i = 0}^n$ via the following rule: $y_i = p^{-1} \cdot S_i$. The vectors $\{y_i\}_{i = 0}^n$ form a \emph{fractional} solution for {\CM} in the following sense. First, for every $0 \leq i \leq n$, the sum of the coordinates of $y_i$ is at most $k$. Second, for every $1 \leq i \leq n$, $y_i \leq y_{i - 1} \vee u$, where the inequality is element-wise. To convert the vectors $\{y_i\}_{i = 0}^n$ into an integral solution $\{\bar{S}_i\}_{i = 0}^n$ for {\CM}, we need an online rounding method.


The rounding we suggest works as follows. For every $1 \leq \ell \leq k$, select a uniformly random value $r_\ell$ from the range $\{p(\ell - 1) + 1, p(\ell - 1) + 2, \dotsc, p\ell\}$. Intuitively, the values $r_\ell$ specify the indexes of the elements from the set $\{u_j\}_{j = 1}^{pk}$ that get into the rounded solution. Later elements can get into the rounded solution only if they ``take the place'' of one of these elements. More formally, we say that an element $u_j \in S_i$ uses place $a$ if there exists a series of indexes: $j_1, j_2, \dotsc, j_m$ such that:
\[
	j = j_1
	\enspace,
	\quad
	u'_{j_i} = u_{j_{i - 1}} \enspace \text{(\ie, $u_{j_i}$ replaced $u_{j_{i-1}}$ in $S_{j_i}$)}
	\quad
	\inConference{\]\[}
	\text{and}
	\quad
	j_m = a
	\enspace.
\]
The element $u_j \in S_i$ gets into $\bar{S}_i$ if and only if it uses a place $a$ equal to $r_\ell$ for some $\ell \in \{1, 2, \dotsc, k\}$. 

\begin{observation} \label{ob:single_use_place}
For every $0 \leq i \leq n$, at most one element of $S_i$ uses every place $a \in \{1, 2, \dotsc, pk\}$.
\end{observation}
\begin{proof}
The observation follows immediately from the definition when $i \leq pk$ because for every $1 \leq j \leq i$, the element $u_j$ uses place $j$. For larger values of $i$ we prove the observation by induction. Assume the claim holds for $i - 1 \geq pk$, and let us prove it for $i$. Since $u_i$ got into $S_i$, the place used by $u_i$ is defined as the place used by $u'_i$. As $S_i = S_{i - 1} + u_i - u'_i$, every place is used at most once also in $S_i$.
\end{proof}

\begin{corollary}
The sets $\{\bar{S}_i\}_{i = 0}^n$ form a feasible online solution for {\CM}.
\end{corollary}
\begin{proof}
Observation~\ref{ob:single_use_place} implies immediately that $|\bar{S}_i| \leq k$ since only $k$ places make elements of $S_i$ using them appear in $\bar{S}_i$. Observe also that the rule whether an element of $S_i$ gets into $\{\bar{S}_i\}$ is independent of $i$. Hence, for every $1 \leq i \leq k$, $\bar{S}_i \subseteq \bar{S}_{i - 1} + u_i$ simply because $S_i \subseteq S_{i - 1} + u_i$.
\end{proof}

Next, we want to show that $f(\bar{S}_n)$ is as good approximation for $f(OPT)$ as $g(S_n)$. This is done by the following lemma.
\begin{lemma}
$\mathbb{E}[f(\bar{S}_n)] \geq g(S_n)$.
\end{lemma}
\begin{proof}
For every $1 \leq \ell \leq k$, let $R_\ell$ be the set of elements of $S_n$ using places from $\{p(\ell - 1) + 1, p(\ell - 1) + 2, p\ell\}$. By Observation~\ref{ob:single_use_place}, the size of $R_\ell$ is at most $p$. Let $\bar{R}_\ell$ be a set containing the single element of $R_\ell$ using the place $r_\ell$, or the empty set if no element uses this place. By definition, $S_n = \bigcup_{\ell = 1}^k \bar{R}_\ell$.

Recall also that given a vector $x \in [0, 1]^\NN$, $\RSet(x)$ is a random set containing every element $u \in R_i$ with probability $x_u$. Using this notation,
\[
	g(S_n)
	=
	\mathbb{E}[f(\RSet(p^{-1} \cdot S_n))]
	=
	\mathbb{E}\left[f\left(\bigcup_{\ell = 1}^k \RSet(p^{-1} \cdot R_\ell)\right)\right]
	\enspace.
\]
Hence, the lemma is equivalent to the inequality:
\begin{equation} \label{eq:equivalent_inequality}
	\mathbb{E}\left[f\left(\bigcup_{\ell = 1}^k \bar{R}_\ell\right)\right]
	\geq
	\mathbb{E}\left[f\left(\bigcup_{\ell = 1}^k \RSet(p^{-1} \cdot R_\ell)\right)\right]
	\enspace.
\end{equation}

As a step toward proving Inequality~\eqref{eq:equivalent_inequality}, fix $1 \leq \ell \leq k$, and let $D_\ell$ be a random subset of $\NN \setminus R_\ell$ with an arbitrary distribution. We use the notation $v_1, v_2, \dotsc, v_{p'}$ to denote the elements of $R_\ell$ ($p' = |R_\ell| \leq p$), and $R_{\ell, i}$ to denote the set $\{v_1, v_2, \dotsc, v_i\} \subseteq R_\ell$. By submodularity,
\begin{align*}
	\inConference{&}
	\mathbb{E}\left[f\left(D_\ell \cup \bar{R}_\ell\right)\right]
	=\inArxiv{{} &}
	\mathbb{E}[f(D_\ell)] + \frac{\sum_{i = 1}^{p'} \mathbb{E}[f(v_i \mid D_\ell)]}{p} \\
	\geq{} &
	\mathbb{E}[f(D_\ell)] + \frac{\sum_{i = 1}^{p'} \mathbb{E}[f(v_i \mid D_\ell \cup \RSet(p^{-1} \cdot R_{\ell, i - 1})]}{p}\inConference{\\}
	=\inConference{{} &}
	\mathbb{E}\left[f\left(D_\ell \cup \RSet(p^{-1} \cdot R_\ell)\right)\right]
	\enspace.
\end{align*}
Inequality~\eqref{eq:equivalent_inequality} follows by repeated applications of the above inequality, once for every $1 \leq \ell \leq k$.
\end{proof}

The first part of Theorem~\ref{th:cardinality_positive} follows by combining the modified Algorithm~\ref{alg:GreedyWithThreshold} with the rounding method described above.


\subsection{Hardness results for {\CM}\inConference{.}} \label{ssc:cardinality_negative}

In this section we prove Theorem~\ref{th:cardinality_negative}. The proof of the theorem is split between two lemmata. Lemma~\ref{le:cardinality_negative_monotone} proves the part of Theorem~\ref{th:cardinality_negative} referring to monotone objectives, and Lemma~\ref{le:cardinality_negative_general} proves the part referring to general objective functions.

\begin{lemma} \label{le:cardinality_negative_monotone}
No deterministic (randomized) online algorithm has a competitive ratio better than $1/2 + \ee$ ($3/4 + \ee$) for {\CM}, even if the objective function is restricted to be monotone.
\end{lemma}
\begin{proof}
Let $k = \lceil (2\ee)^{-1} \rceil$ and consider the ground set $\NN = \{u_i\}_{i = 1}^k \cup \{v_i\}_{i = 1}^k \cup \{w\}$ and the monotone submodular function $f : 2^\NN \rightarrow \mathbb{R}^+$ defined as follows:
\[
	f(S)
	=
	|S \cap \{u_i\}_{i = 1}^k| + \min\{k, |S \cap \{v_i\}_{i = 1}^k| + k \cdot |S \cap \{w\}|\}
	\enspace.
\]

Intuitively, $f$ gains a ``point'' for every element of the forms $u_i$ and $v_i$. The element $w$ gives $k$ ``points'', but these are the same points of the elements of the form $v_i$, \ie, $v_i$ gives no ``points'' once $w$ is in the solution.

Fix an algorithm $ALG$, and set the element $w$ to be revealed after all the other elements of $\NN$. Observe that as long as $w$ is not revealed, there is no way for $ALG$ to distinguish between the elements of $\{u_i\}_{i = 1}^k$ and $\{v_i\}_{i = 1}^k$. Thus, if $ALG$ is deterministic, the set $S_{2k}$ is determined solely based on the order in which the elements are revealed. By choosing the right order, one can guarantee that $S_{2k} \subseteq \{v_i\}_{i = 1}^k$, which implies that $S_{2k + 1}$ is a subset of $\{v_i\}_{i = 1}^k \cup \{w\}$, and thus, has a value of at most $k$. On the other hand, the optimal solution is $\{u_i\}_{i = 1}^{k - 1} \cup \{w\}$, and its value is $2k - 1$. Hence, the competitive ratio of $ALG$ is at most:
\[
	\frac{k}{2k - 1}
	\leq
	\frac{k + 1}{2k}
	\leq
	\frac{1}{2} + \ee
	\enspace.
\]

If $ALG$ is randomized, $S_{2k}$ depends also on the random choices of $ALG$. However, by symmetry, it is still possible to choose a revelation order for the elements of $\{u_i\}_{i = 1}^k$ and $\{v_i\}_{i = 1}^k$ guaranteeing that in expectation $S_{2k}$ contains no more than $k/2$ elements of $\{u_i\}_{i = 1}^k$. Thus, the expected value of $f(S_{2k + 1})$ can be upper bounded by $3k/2$, and the competitive ratio of $ALG$ is no more than:
\[
	\frac{3k/2}{2k - 1}
	\leq
	\frac{3k/2 + 1}{2k}
	\leq
	\frac{3}{4} + \ee
	\enspace.
	\qedhere
\]
\end{proof}

\begin{lemma} \label{le:cardinality_negative_general}
No online algorithm has a competitive ratio better than $1/2 + \ee$ for {\CM}.
\end{lemma}
\begin{proof}
Intuitively, the proof of Lemma~\ref{le:cardinality_negative_general} provides a weaker bound for randomized algorithms because the number of elements of the form $v_i$ in the ground set is quite small. In this proof we fix this problem. Let $k = \lceil \ee^{-1} \rceil$ and $h = 2k \cdot \lceil \ee^{-1} \rceil$. Consider the ground set $\NN = \{u_i\}_{i = 1}^k \cup \{v_i\}_{i = 1}^h \cup \{w\}$ and the non-negative submodular function $f : 2^\NN \rightarrow \mathbb{R}^+$ defined as follows:
\[
	f(S)
	=
	\begin{cases}
		|S| & \text{if $w \not \in S$} \enspace,\\
		k + |S \cap \{u_i\}_{i = 1}^k|  & \text{otherwise} \enspace.
	\end{cases}
\]

Fix an algorithm $ALG$, and set the element $w$ to be revealed after all the other elements of $\NN$. Observe that as long as $w$ is not revealed, there is no way for $ALG$ to distinguish between the elements of $\{u_i\}_{i = 1}^k$ and $\{v_i\}_{i = 1}^h$. Hence, by symmetry, it is possible to choose a revelation order for these elements guaranteeing that in expectation $S_{k + h}$ contains no more than $k^2/(k + h)$ elements of $\{u_i\}_{i = 1}^k$. The final solution of $ALG$ is a subset of $S_{k + h} + w$. Since the size of $S_{k + h}$ is at most $k$, no element of $S_{k + h} + w$ has a negative marginal contribution, and thus, the value of $ALG$'s solution can be upper bounded by $f(S_{k + h} + w)$. Hence, the total expected value $ALG$ achieves is at most:
\[
	\mathbb{E}[f(S_{k + h} + w)]
	=
	k + \mathbb{E}[|S_{k + h} \cap \{u_i\}_{i = 1}^k|]
	\leq
	k + \frac{k^2}{k + h}
	\enspace.
\]

On the other hand, the optimal solution is $\{u_i\}_{i = 1}^{k - 1} \cup \{w\}$, and its value is $2k - 1$. Hence, the competitive ratio of $ALG$ is at most:
\[
	\frac{k + \frac{k^2}{k + h}}{2k - 1}
	\leq
	\frac{k + 1}{2k} + \frac{k}{k + h}
	\leq
	\frac{1}{2} + \ee
	\enspace.
	\qedhere
\]
\end{proof}

\apptocmd{\sloppy}{\hbadness 10000\relax}{}{}
\bibliographystyle{plain}
\bibliography{submodular}

\appendix
\section{Proof of Theorem~\ref{th:unconstrained_no_preemption}} \label{sec:unconstrained_no_preemption}

In this section we prove Theorem~\ref{th:unconstrained_no_preemption}. Namely, we need to show that no algorithm using no preemption is $(1/4 + \ee)$-competitive for any constant $\ee > 0$ even when $f$ is guaranteed to be a cut function of a weighted directed graph and there is no constraint on the set of elements that can be picked. Assume for the sake of contradiction that there exists such an algorithm $ALG$ for some, fixed, $\ee > 0$. We design an adversary that constructs an instance for $ALG$ that leads to a contradiction. The adversary uses a ground set $\NN$ with $n = n(\ee)$ elements (where $n(\ee)$ is function that depends only on $\ee$ and is described later).

Each time an element $u_i$ is revealed the adversary uses the properties of $ALG$ to choose the weighted set $E_i$ of arcs that leaves $u_i$. The objective function $f$ is defined as the weighted cut function of the graph $(\NN, \bigcup_{i = 1}^n E_i)$. Notice that any query of $f$ made by $ALG$ before $u_i$ is revealed does not depend on $E_i$, and thus, the behavior of $ALG$ up to the point when $u_i$ is revealed can be described before $E_i, E_{i + 1}, \dotsc, E_n$ are determined. An exact description of the adversary is given as Algorithm~\ref{alg:Adversary}. For ease of notation, we denote by $\NN_i$ the set $\{u_j \mid 1 \leq j \leq i\}$ for every $0 \leq i \leq n$.

\SetKwData{Mode}{Mode}
\begin{algorithm*}[h!t]
\caption{\textsf{Adversary for $ALG$}} \label{alg:Adversary}
Let $m \gets 1 + 4 / \ee$, $n_0 \gets \lceil 8/\ee \rceil$ and $n = \lceil e^{8/\ee} \cdot n_0\rceil$. \\
Let $\Mode \gets A$.

\BlankLine

\For{$i$ = $1$ to $n$}
{
	\If{$\Mode = A$}
	{
		Reveal element $u_i$ with arcs of weight $m^i$ going to all nodes revealed so far.\\
	}
	\Else(\tcp*[h]{If $\Mode = B$})
	{
		Reveal element $u_i$ with no exiting arcs.
	}
	
	\BlankLine
	
	\For{every set $S \subseteq V_{i - 1}$}
	{
		Let $A_{i, S}$ be the event that the algorithm accepted exactly the elements of $S$ before $u_i$ is revealed.\\
		Let $p_{i, S}$ be the probability of $A_{i, S}$.\\
		Let $q_{i, S}$ be the probability $u_i$ is accepted given $A_{i, S}$.
	}
	\If{$i \geq n_0$ and $\sum_{S \subseteq N_{i - 1}} \left[p_{i, S} \cdot q_{i, S} \cdot \left(1 - \frac{|S|}{i - 1}\right)\right] < 1/4 + \ee/4$}
	{
		$\Mode \gets B$.\\
	}
}
\end{algorithm*}

Observe that Algorithm~\ref{alg:Adversary} operates in two modes. In mode $A$ every revealed element has an arc going to every previously revealed element. In mode $B$ the revealed elements have no exiting arcs (\ie, their output degrees are $0$). The adversary starts in mode $A$ and then switches permanently to mode $B$ when some condition holds for the first time.

\begin{lemma}
The adversary given by Algorithm~\ref{alg:Adversary} never switches to mode $B$.
\end{lemma}
\begin{proof}
Assume for the sake of contradiction that the adversary switches to mode $B$ immediately after element $u_i$ is revealed. Let us upper bound the weight of the cut $C$ produced by $ALG$ in this case. If $ALG$ accepted a set $S \subseteq S_{i - 1}$ of elements before $u_i$ is revealed, then the expected number of arcs leaving $u_i$ crossing $C$ is $(i - 1) \cdot q_{i, S} \cdot \left(1 - \frac{|S|}{i - 1}\right)$. By the linearity of the expectation, the expected number of arcs leaving $u_i$ crossing $C$ is:
\begin{align*}
	(i - 1) \cdot \sum_{S \subseteq \NN_{i - 1}} \inConference{\hspace{1in}&\hspace{-1in}} \left[p_{i, S} \cdot q_{i, S} \cdot \left(1 - \frac{|S|}{i - 1}\right)\right] \inConference{\\}
	<{} &
	(i - 1) \cdot (1/4 + \ee/4)
	\enspace.
\end{align*}

The total weight of all arcs leaving elements other than $u_i$ is:
\begin{align*}
	\sum_{j = 1}^{i - 1} [(j - 1) \cdot m^j]
	\leq{} &
	(i - 1) \inArxiv{\cdot} \sum_{j = 1}^{i - 1} m^j
	\leq
	(i - 1) \inArxiv{\cdot} \sum_{j = 0}^\infty m^{i - 1 - j}\inConference{\\}
	=\inConference{{} &}
	(i - 1) \cdot \frac{m^i}{m - 1}
	=
	(i - 1)m^i \cdot \frac{\ee}{4}
	\enspace.
\end{align*}
Hence, even if all the above arcs cross $C$, the total weight of $C$ is still less than $(i - 1)m^i \cdot (1/4 + \ee/2)$. On the other hand, the optimal solution may accept $u_i$ and no other elements, which results in a cut of weight $(i - 1)m^i$. Hence, the competitive ratio of $ALG$ cannot be $1/4 + \ee$.
\end{proof}

Let us denote $t_{i, S} = 1 - \frac{|S|}{i - 1}$. Using this notation the previous lemma implies that for every $n_0 \leq i \leq n$:
\[
	\sum_{S \subseteq \NN_{i - 1}} \left(p_{i, S} \cdot q_{i, S} \cdot t_{i, S}\right) \geq 1/4 + \ee/2
	\enspace.
\]
Consider the potential function\inArxiv{ $\Phi(i) = \sum_{S \subseteq \NN_{i - 1}} \left(p_{i, S} \cdot t_{i, S}^2 \right)$.}\inConference{:
\[
	\Phi(i)
	=
	\sum_{S \subseteq \NN_{i - 1}} \left(p_{i, S} \cdot t_{i, S}^2 \right)
	\enspace.
\]
}

\begin{observation} \label{ob:diff}
For every $1 \leq i \leq n - 1$,
\begin{align*}
	\inConference{&}
	\Phi(i + 1) - \Phi(i)\inConference{\\}
	={} &
	\sum_{S \subseteq \NN_{i - 1}} \left(p_{i, S} \cdot \left[\frac{[(i - 1)t_{i, S} + (1 - q_{i, S})]^2}{i^2} - t_{i, S}^2\right] \right)
	\inArxiv{\\
	={} &
	\sum_{S \subseteq \NN_{i - 1}} \left(p_{i, S} \cdot \frac{(1 - 2i)t_{i, S}^2 + 2(i - 1)t_{i, S}(1 - q_{i, S}) + (1 - q_{i, S})^2}{i^2} \right)
	}
	\enspace.
\end{align*}
\end{observation}
\begin{proof}
Let $B_i$ be the expected set of elements accepted by $ALG$ immediately after $u_i$ is revealed. Notice that:
\begin{align*}
	\Phi(i + 1)
	={} &
	\mathbb{E}\left[\left(1 - \frac{|B_i|}{i}\right)^2\right] \inConference{\\}
	=\inConference{{} &}
	\sum_{S \subseteq \NN_{i - 1}} \left(p_{i, S} \cdot \mathbb{E}\left[\left(1 - \frac{|B_i|}{i}\right)^2 ~\middle|~ A_{i, S}\right]\right)
	\enspace.
\end{align*}
Observe also that $\mathbb{E}[|B_i| \mid A_{i, S}] = \mathbb{E}[|S| + q_{i, S} \mid A_{i, S}]$. Hence, by the linearity of the expectation.
\begin{align*}
	\Phi(i +\inConference{&} 1)\inConference{\\}
	={} &
	\sum_{S \subseteq \NN_{i - 1}} \left(p_{i, S} \cdot \mathbb{E}\left[\left(1 - \frac{|S| + q_{i, S}}{i}\right)^2 ~\middle|~ A_{i, S}\right]\right)\\
	={} &
	\sum_{S \subseteq \NN_{i - 1}} \left(p_{i, S} \cdot \left(1 - \frac{(1 - t_{i, S})(i - 1) + q_{i, S}}{i}\right)^2 \right)\\
	={} &
	\sum_{S \subseteq \NN_{i - 1}} \left(p_{i, S} \cdot \left(\frac{t_{i, S}(i - 1) + (1 - q_{i, S})}{i}\right)^2 \right)
	\enspace.
\end{align*}
The observation follows by combining the above equality with the definition of $\Phi(i)$.
\end{proof}

Using the above observation, we can now prove the following lemma.
\begin{lemma} \label{le:diff_bound}
For every $n_0 \leq i \leq n - 1$, $\Phi(i + 1) - \Phi(i) \leq -\ee/(4i)$.
\end{lemma}
\begin{proof}
By Observation~\ref{ob:diff},
\begin{align*}
	\inConference{&}
	\Phi(i + 1) - \Phi(i) \inConference{\\}
	={} &
	\inConference{
		\sum_{S \subseteq \NN_{i - 1}} \left(p_{i, S} \cdot \left[\frac{[(i - 1)t_{i, S} + (1 - q_{i, S})]^2}{i^2} - t_{i, S}^2\right] \right)\\
		={} &
		\sum_{S \subseteq \NN_{i - 1}} \left(p_{i, S} \cdot \frac{(1 - 2i)t_{i, S}^2}{i^2} \right. \\ &+ \left. \frac{2(i - 1)t_{i, S}(1 - q_{i, S}) + (1 - q_{i, S})^2}{i^2} \right)\\
	}
	\inArxiv{
		\sum_{S \subseteq \NN_{i - 1}} \left(p_{i, S} \cdot \frac{(1 - 2i)t_{i, S}^2 + 2(i - 1)t_{i, S}(1 - q_{i, S}) + (1 - q_{i, S})^2}{i^2} \right)\\
	}
	\leq{} &
	2 \cdot \sum_{S \subseteq \NN_{i - 1}} \left(p_{i, S} \cdot \frac{- i \cdot t^2_{i, S} + i \cdot t_{i, S} \cdot (1 - q_{i, S}) + 1}{i^2} \right)\\
	={} &
	\frac{2}{i} \cdot \sum_{S \subseteq \NN_{i - 1}} \left(p_{i, S} \cdot t_{i, S} \cdot (1 - q_{i, S} - t_{i, S}) \right) + \frac{2}{i^2}
	\enspace.
\end{align*}
Notice that $t_{i, S} \cdot (1 - q_{i, S} - t_{i, S}) = t_{i, S} \cdot (1 - t_{i, S}) - t_{i, S} \cdot q_{i, S} \leq 1/4 - t_{i, S} \cdot q_{i, S}$. Plugging this inequality into the previous one yields:
\begin{align*}
	\Phi(i + 1) - \inConference{&}\Phi(i)\inConference{\\}
	\leq{} &
	\frac{2}{i} \cdot \sum_{S \subseteq \NN_{i - 1}} \left(p_{i, S} \cdot (1/4 - t_{i, S} \cdot q_{i, S} ) \right) + \frac{2}{i^2}\inConference{\\}
	=\inConference{{} &}
	\frac{2}{i} \cdot \left[1/4 - \sum_{S \subseteq \NN_{i - 1}} p_{i, S} \cdot t_{i, S} \cdot q_{i, S} \right] + \frac{2}{i^2}\\
	\leq{} &
	\frac{2}{i} \cdot \left[1/4 - (1/4 + \ee/4) \right] + \frac{2}{i^2}
	=
	- \frac{\ee}{2i} + \frac{2}{i^2}\inConference{\\}
	\leq\inConference{{} &}
	- \frac{\ee}{2i} + \frac{\ee}{4i}
	=
	- \frac{\ee}{4i}
	\enspace,
\end{align*}
where the third inequality holds since $i \geq n_0 \geq 8/\ee$.
\end{proof}

\begin{corollary} \label{co:total_loss}
$\Phi(n) \leq \Phi(n_0) - 2$.
\end{corollary}
\begin{proof}
By Lemma~\ref{le:diff_bound},
\begin{align*}
	\Phi(n) - \Phi(n_0)
	={} &
	\sum_{i = n_0}^{n - 1} [\Phi(i + 1) - \Phi(i)]
	\leq
	- \sum_{i = n_0}^{n - 1} \frac{\ee}{4i}\\
	\leq{} &
	-\frac{\ee}{4} \cdot \int_{n_0}^n \frac{dx}{x}
	=
	-\frac{\ee}{4} \cdot \ln(n / n_0)\inConference{\\}
	\leq\inConference{{} &}
	-\frac{\ee}{4} \cdot \ln(n_0 \cdot e^{8/\ee} / n_0)
	=
	-2
	\enspace.
	\qedhere
\end{align*}
\end{proof}

Corollary~\ref{co:total_loss} leads to an immediate contradiction since $\Phi$ can only take values from the range $[0, 1]$.
\inArxiv{\section{A Polynomial Time Implementation of Algorithm~\ref{alg:IncreaseFromZero}} \label{app:unconstrained_polynomial}

In this section we prove the second part of Theorem~\ref{th:unconsrained_positive} by giving a polynomial time implementation of Algorithm~\ref{alg:IncreaseFromZero} having the same competitive ratio up to an arbitrary constant $\ee > 0$. Our implementation is given as Algorithm~\ref{alg:IncreaseFromZero_Polynomial}. This algorithms uses the notation $T_\theta(x)$, where $x$ is a vector and $\theta \in [0, 1]$. This notation stands for $\{u \mid x_u \geq \theta\}$, \ie, the set of all elements whose coordinate in $x$ is at least $\theta$.

\begin{algorithm}[h!t]
\caption{\textsf{Marginal Choice - Polynomial Time Implementation}} \label{alg:IncreaseFromZero_Polynomial}
\ForEach{element $u_i$ revealed}
{
	Choose a uniformly random threshold $\theta_i \in [0, 1]$.\\
	Let $\NN_i \gets \{u_1, u_2, \dotsc, u_i\}$.\\

	\BlankLine

	Let $\VV_i$ be a collections of $\lceil 3072 \cdot \ee^{-2}i^4 \ln[2\ee^{-1}(4i)^2] \rceil$ random vectors from $[0, 1]^{\NN_i}$.\\
	\For{$j$ = $1$ \KwTo $i - 1$}
	{
		\For{every vector $v \in \VV_j$}
		{
			Extend $v$ to be a vector of $[0, 1]^{\NN_i}$ by setting $v_{u_i}$ to be an independent uniformly random value from $[0, 1]$.\\
		}
	}
	
	\BlankLine
	
	\For{$j$ = $1$ \KwTo $i$}
	{
		$m_j \gets |\VV_j|^{-1} \cdot \sum_{v \in \VV_j} [f(T_{\theta_j}(v) + u_j) - f(T_{\theta_j}(v) - u_j)]$. \tcc{$m_j \approx \partial_{u_j} F(\theta_j \cdot \NN_i)$.}
	}
	Let $S_i \gets \{u_j \in \NN_i \mid m_j \geq 0\}$.
}
\end{algorithm}

Our first objective is to show that Algorithm~\ref{alg:IncreaseFromZero_Polynomial} is an online algorithm according to {\UM}.

\begin{lemma}
For every $1 \leq i \leq n$, $S_i \subseteq S_{i - 1} + u_i$.
\end{lemma}
\begin{proof}
For $i = 1$ the claim is trivial since $S_1$ can contain only $u_1$. Thus, we assume from now on $i > 1$. For every $1 \leq j \leq i - 1$, let $m_j$ and $\VV_j$ denote the value $m_j$ and the set $\VV_j$ after iteration $i - 1$, and let $m'_j$ and $\VV'_j$ denote these entities after iteration $i$.

By definition, $S_i$ contains only elements of $\NN_i$. Fix an element $u_j \in \NN_i - u_i$. Then:
\begin{align*}
	u_j \in S_i
	\Rightarrow{} &
	\frac{\sum_{v \in \VV'_j} [f(T_{\theta_j}(v) + u_j) - f(T_{\theta_j}(v) - u_j)]}{|\VV_j|} = m'_j \geq 0\\
	\Rightarrow{} &
	\frac{\sum_{v \in \VV_j} [f(T_{\theta_j}(v) + u_j) - f(T_{\theta_j}(v) - u_j)]}{|\VV_j|} = m_j \geq 0
	\Rightarrow
	u_j \in S_{i - 1}
	\enspace,
\end{align*}
where the second derivation follows from submodularity since $\VV'_j$ is produced from $\VV_j$ by extending vectors.
\end{proof}

To bound the loss in the competitive ratio of Algorithm~\ref{alg:IncreaseFromZero_Polynomial}, we need the following claims.

\begin{observation}
For every $1 \leq i \leq n$, each set in $\VV_i$ is an independent sample of $[0, 1]^\NN$.
\end{observation}
\begin{proof}
Follows immediately by the way Algorithm~\ref{alg:IncreaseFromZero_Polynomial} constructs $\VV_i$ and the way this collection is updated when new elements arrive.
\end{proof}

\begin{lemma} \label{lem:concentration_bound}
Let $X_1, X_2, \dotsc, X_\ell$ be independent random variables such that for each $i$, $X_i \in [-1, 1]$. Let $X = \frac{1}{\ell} \sum_{i = 1}^\ell X_i$ and $\mu = \mathbb{E}[X]$. 
Then
\[
	\Pr[X > \mu + \alpha] \leq e^{-\frac{\alpha^2\ell}{12}}
	\qquad \mbox{and} \qquad
	\Pr[X < \mu - \alpha] \leq e^{-\frac{\alpha^2\ell}{8}}
	\enspace,
\]
for every $\alpha > 0$.
\end{lemma}
\begin{proof}
For every $1 \leq i \leq \ell$, let $Y_i = (1 + X_i) / 2$. Additionally, let $Y = \sum_{i = 1}^\ell Y_i$.  Clearly, $Y_i \in [0, 1]$. Hence, by the Chernoff bound:
\begin{align*}
	\Pr[X > \mu + \alpha]
	={} &
	\Pr[2Y/\ell - 1 > 2\mathbb{E}[Y]/\ell - 1 + \alpha]
	=
	\Pr\left[Y > \mathbb{E}[Y] + \frac{\alpha\ell}{2}\right]\\
	\leq{} &
	e^{-\frac{\mathbb{E}[Y](\alpha\ell / (2 \cdot \mathbb{E}[Y]))^2}{3}}
	=
	e^{-\frac{\alpha^2\ell^2}{12 \cdot \mathbb{E}[Y]}}
	\leq
	e^{-\frac{\alpha^2\ell}{12}}
	\enspace,
\end{align*}
and
\begin{align*}
	\Pr[X < \mu - \alpha]
	={} &
	\Pr[2Y/\ell - 1 < 2\mathbb{E}[Y]/\ell - 1 - \alpha]
	=
	\Pr\left[Y < \mathbb{E}[Y] - \frac{\alpha\ell}{2}\right]\\
	\leq{} &
	e^{-\frac{\mathbb{E}[Y](\alpha\ell / (2 \cdot \mathbb{E}[Y]))^2}{2}}
	=
	e^{-\frac{\alpha^2\ell^2}{8 \cdot \mathbb{E}[Y]}}
	\leq
	e^{-\frac{\alpha^2\ell}{8}}
	\enspace.
	\qedhere
\end{align*}
\end{proof}

Let $m_i(z)$ be the value that $m_i$ gets after the last iteration of Algorithm~\ref{alg:IncreaseFromZero_Polynomial} if the random variable $\theta_i$ is $z$.

\begin{lemma} \label{lem:wrong_estimate_probability}
For every $1 \leq i \leq n$ and $z \in [0, 1]$, $\Pr[|m_i(z) - \partial_{u_i} F(z \cdot \NN)| > \ee (4i)^{-2} \cdot f(OPT)] \leq \ee (4i)^{-2}$.
\end{lemma}
\begin{proof}
Observe that $m_i$ is obtained by averaging $|\VV_i|$ independent samples of $f(R(z \cdot \NN) + u_i) - f(R(z \cdot \NN) - u_i)$, and the expectation of the last expression is $\partial_{u_i} F(z \cdot \NN)$. Additionally, observe that every sample belongs to the range $[-f(OPT), f(OPT)]$ since $f(OPT) = \max_{S \subseteq \NN} f(S)$. Hence, by Lemma~\ref{lem:concentration_bound},
\[
	\Pr[|m_i(z) - \partial_{u_i} F(z \cdot \NN)| > \ee (4i)^{-2} \cdot f(OPT)]
	\leq
	2e^{\frac{-|\VV_i| \cdot [\ee (4i)^{-2}]^2}{12}}
	\leq
	2e^{-\ln[2\ee^{-1}(4i)^2]}
	=
	\ee (4i)^{-2}
	\enspace.
	\qedhere
\]
\end{proof}

Fix the vectors $\VV = (\VV_1, \VV_2, \dotsc, \VV_n)$. For every element $u_i \in \NN$, $m_i(z)$ is continuous non-increasing function of $z$ by submodularity. Hence, by the intermediate value theorem, one of the following must hold: $m_i(z)$ is always positive in the range $[0, 1]$, $m_i(z)$ is always negative in the range $[0, 1]$ or $m_i(z)$ has at least one root $z_0 \in [0, 1]$. In the last case, the set $I_0 \subseteq [0, 1]$ of the roots of $m_i(z)$ is non-empty. Moreover, by the monotonicity and continuity of $m_i(z)$, $I_0$ is a closed range. Using these observations, we define a vector $y(\VV) \in [0, 1]^\NN$ as follows:
\[
	y(\VV)_{u_i}
	=
	\begin{cases}
		0 & \text{if $m_i(0) < 0$} \enspace, \\
		1 & \text{if $m_i(1) > 0$} \enspace, \\
		\max I_0 & \text{otherwise} \enspace.
	\end{cases}
\]

Let $\EE_i$ be the event that $\partial_{u_i}F(z \cdot \NN) \geq -\ee (4i)^{-2} \cdot f(OPT)$ for every $z < y(\VV)_{u_i}$ and $\partial_{u_i}F(z \cdot \NN) \leq \ee (4i)^{-2} \cdot f(OPT)$ for every $z > y(\VV)_{u_i}$.

\begin{lemma} \label{lem:coordinate_y_good}
For every $1 \leq i \leq n$, the event $\EE_i$ holds with probability at least $1 - 2\ee (4i)^{-2}$.
\end{lemma}
\begin{proof}
Let us first bound the probability that $\partial_{u_i}F(z \cdot \NN) \geq -\ee (4i)^{-2} \cdot f(OPT)$ for every $z < y(\VV)_{u_i}$. If $\partial_{u_i}F(z \cdot \NN)$ is always at least $-\ee (4i)^{-2} \cdot f(OPT)$, then there is nothing to prove. Otherwise, if $\partial_{u_i}F(z \cdot \NN)$ is always at most $-\ee (4i)^{-2} \cdot f(OPT)$, then $\partial_{u_i} F(\varnothing) \leq 0$, which guarantees that $y(\VV)_{u_i} = 0$, and thus, makes the claim that we want to prove void. Hence, it is enough to consider the case in which $\partial_{u_i}F(z \cdot \NN)$ is at least $-\ee (4i)^{-2} \cdot f(OPT)$ for some value of $z$ and at most $-\ee (4i)^{-2} \cdot f(OPT)$ for another value of $z$. In this case the intermediate value theorem and the continuity of $\partial_{u_i}F(z \cdot \NN)$ imply the existence of a value $y'_i$ such that $\partial_{u_i}F(y'_i \cdot \NN) = -\ee (4i)^{-2} \cdot f(OPT)$. By Lemma~\ref{lem:wrong_estimate_probability}, $m_i(y'_i) \leq 0$ with probability at least $1 - \ee (4i)^{-2}$, which implies $y(\VV)_{u_i} \leq y'_i$, and thus:
\[
	\partial_{u_i}F(y(\VV)_{u_i} \cdot \NN)
	\geq
	\partial_{u_i}F(y'_i \cdot \NN)
	=
	-\ee (4i)^{-2} \cdot f(OPT)
	\enspace.
\]
Observe that the last inequality implies that $\partial_{u_i}F(z \cdot \NN) \geq -\ee (4i)^{-2} \cdot f(OPT)$ for every $z < y(\VV)_{u_i}$.

A similar argument shows that $\partial_{u_i}F(z \cdot \NN) \leq \ee (4i)^{-2} \cdot f(OPT)$ for every $z > y(\VV)_{u_i}$ with probability at least $1 - \ee (4i)^{-2}$. The lemma now follow by the union bound.
\end{proof}

Let $\EE$ be the event that $\EE_i$ holds for every $1 \leq i \leq n$ at the same time.
\begin{corollary} \label{cor:EE_probability}
$\EE$ happens with probability at least $1 - \ee/2$.
\end{corollary}
\begin{proof}
Combining Lemma~\ref{lem:coordinate_y_good} and the union bound gives:
\[
	\Pr[\EE]
	\geq
	1 - \sum_{i = 1}^n 2\ee (4i)^{-2}
	\geq
	1 - \frac{\ee}{8} - \int_1^n 2\ee (4x)^{-2}dx
	=
	1 - \frac{\ee}{8} + \left.\frac{\ee}{8x}\right|^n_1
	\geq
	1 - \frac{\ee}{2}
	\enspace.
	\qedhere
\]
\end{proof}

It is important to notice that the event $\EE$ depends only on the vectors $\VV$. Thus, fixing the vectors $\VV$ does not affect the distribution of $\theta_1, \theta_2, \dotsc, \theta_n$. Next, we analyze the competitive ratio of Algorithm~\ref{alg:IncreaseFromZero_Polynomial} assuming $\VV$ is fixed in a way that makes the event $\EE$ hold. The next observation corresponds to Observation~\ref{obs:y_alg_equivalent}.

\begin{observation} \label{obs:y_alg_equivalent_Polynomial}
For every given vectors $\VV$, every element $u_i \in \NN$ belongs to $S_n$ with probability $y(\VV)_{u_i}$, independently. Hence, $\mathbb{E}[f(S_n)] = F(y(\VV))$.
\end{observation}
\begin{proof}
An element $u_i \in \NN$ belongs to $S_n$ if and only if $m_i(\theta_i) \geq 0$, which is equivalent to the condition $\theta_i \leq y(\VV)_{u_i}$ since $m_i$ is a non-increasing function. Clearly, the last condition happens with probability $y(\VV)_{u_i}$, and is independent for different elements.
\end{proof}

The last observation implies that analyzing Algorithm~\ref{alg:IncreaseFromZero_Polynomial} is equivalent to lower bounding $F(y(\VV))$. The next lemma and corollary correspond to Lemma~\ref{le:all_worse_than_algorithm} and Corollary~\ref{cor:unconstrained_ratio}, respectively.

\begin{lemma} \label{le:all_worse_than_algorithm_Polynomial}
For every given vectors $\VV$ for which the event $\EE$ holds and for every $\lambda \in [0, 1]$, $F(y(\VV) \wedge (\lambda \cdot \NN)) \geq F(\lambda \cdot \NN) - \ee \cdot \sum_{i = 1}^n (4i)^{-2} \cdot f(OPT)$.
\end{lemma}
\begin{proof}
For ease of notation, we use $y$ as a shorthand for $y(\VV)$ in this proof. Observe that:
\begin{align*}
	F(y \wedge (\lambda \cdot \NN))
	={} &
	f(\varnothing) + \int_0^\lambda \frac{d F(y \wedge (z \cdot \NN))}{d z} dz \inConference{\\}
	=\inConference{{} &}
	f(\varnothing) + \int_0^\lambda \sum_{\substack{u \in \NN \\ z \leq y_u}} \partial_u F(y \wedge (z \cdot \NN)) dz
	\enspace,
\end{align*}
where the second equality is due to the chain rule. By submodularity and the observation that $\EE$ implies $\partial_{u_i} F(z \cdot \NN) \leq \ee (4i)^{-2} \cdot f(OPT)$ for every $z > y_{u_i}$, we get:
\begin{align*}
	\sum_{\substack{u \in \NN \\ z \leq y_u}} \partial_u F(y \wedge (z \cdot \NN))
	\geq{} &
	\sum_{\substack{u \in \NN \\ z \leq y_u}} \partial_u F(z \cdot \NN)\inConference{\\}
	\geq\inConference{{} &}
	\sum_{u \in \NN} \partial_u F(z \cdot \NN) - \ee \cdot \sum_{i = 1}^n (4i)^{-2} \cdot f(OPT)
	\enspace.
\end{align*}
Combining the above equality and inequality, and using the chain rule again, gives:
\begin{align*}
	F(y \wedge (\lambda \cdot \NN))
	\geq{} &
	f(\varnothing) + \int_0^\lambda \left[\sum_{u \in \NN} \partial_u F(z \cdot \NN)\right] dz - \ee \cdot \sum_{i = 1}^n (4i)^{-2} \cdot f(OPT)\\
	={} &
	f(\varnothing) + \int_0^\lambda \frac{d F(z \cdot \NN)}{d z} dz - \ee \cdot \sum_{i = 1}^n (4i)^{-2} \cdot f(OPT)\\
	={} &
	F(\lambda \cdot \NN) - \ee \cdot \sum_{i = 1}^n (4i)^{-2} \cdot f(OPT)
	\enspace.
	\qedhere
\end{align*}
\end{proof}

\begin{corollary} \label{cor:unconstrained_ratio_Polynomial}
For every given vectors $\VV$ for which the event $\EE$ holds, $F(y(\VV)) \geq (e^{-1} - \ee/2) \cdot F(OPT)$.
\end{corollary}
\begin{proof}
Again, we use $y$ as a shorthand for $y(\VV)$ in this proof. Fix two values $0 \leq a \leq b \leq 1$. By the chain rule:
\begin{align} \label{eq:integral_form_Polynomial}
	\inConference{&}F(y \wedge (b \cdot \NN)) -  F(y \wedge (a \cdot \NN))\inConference{\\}
	={} &
	\int_a^b \frac{d F(y \wedge (z \cdot \NN))}{dz} dz \inConference{\nonumber} \\
	={} &
	\int_a^b \sum_{\substack{u \in \NN \\ z \leq y_u}} \partial_u F(y \wedge (z \cdot \NN)) dz\inConference{\nonumber \\}
	\geq\inConference{{} &}
	\int_a^b \sum_{\substack{u \in \NN \\ z \leq y_u}} \partial_u F(z \cdot \NN) dz \nonumber
	\enspace,
\end{align}
where the inequality follows from submodularity. We use two ways to lower bound the rightmost hand side of the above inequality. First, observe that $\partial_{u_i} F(z \cdot \NN) \geq -\ee (4i)^{-2} \cdot f(OPT)$ whenever $z < y_{u_i}$. Thus, $F(y \wedge (b \cdot \NN)) - F(y \wedge (a \cdot \NN)) \geq -\ee(b -a) \cdot \sum_{i = 1}^n (4i)^{-2} \cdot f(OPT)$. Additionally, $\partial_u F(z \cdot \NN) \leq \ee (4i)^{-2} \cdot f(OPT)$ whenever $z > y_u$ and $\partial_u F(z \cdot \NN) \geq -\ee (4i)^{-2} \cdot f(OPT)$ whenever $z < y_u$. This allows us to derive the second lower bound, which holds whenever $z$ is not equal to any coordinate of $y$ (\ie, for every value of $z$ except for a finite set of values):
\begin{align*}
	\sum_{\substack{u \in \NN \\ z \leq y_u}} \partial_u F(z \cdot \NN)
	\geq{} &
	\sum_{u \in OPT} \partial_u F(z \cdot \NN) - \ee \cdot \sum_{i = 1}^n (4i)^{-2} \cdot f(OPT)\\
	={} &
	\frac{\sum_{u \in OPT} [F(u \vee (z \cdot \NN)) - F(z \cdot \NN)]}{1 - z} - \ee \cdot \sum_{i = 1}^n (4i)^{-2} \cdot f(OPT)\\
	\geq{} &
	\frac{F(OPT \vee (z \cdot \NN)) - F(z \cdot \NN)}{1 - z} - \ee \cdot \sum_{i = 1}^n (4i)^{-2} \cdot f(OPT)\\
	\geq{} &
	f(OPT) - \frac{F(y \wedge (z \cdot \NN))}{1 - z} - \ee \cdot \left(1 + \frac{1}{1 - b}\right) \cdot \sum_{i = 1}^n (4i)^{-2} \cdot f(OPT)
	\enspace,
\end{align*}
where the last inequality follows from Lemmata~\ref{le:all_worse_than_algorithm_Polynomial} and~\ref{le:y_max_bound}. Plugging this lower bound and $a = 0$ into~\eqref{eq:integral_form_Polynomial} gives:
\begin{align*}
	F(y \wedge (b \cdot& \NN)) \\
	\geq{} &
	f(\varnothing) + \int_0^b \left[f(OPT) - \frac{F(y \wedge (z \cdot \NN))}{1 - z} - \ee \cdot \left(1 + \frac{1}{1 - b}\right) \cdot \sum_{i = 1}^n (4i)^{-2} \cdot f(OPT)\right] dz
	\enspace.
\end{align*}
The solution of this differential equation is:
\[
	F(y \wedge (b \cdot \NN))
	\geq
	-(1 - b) \cdot \ln (1 - b) \cdot \left(1 - \ee \cdot \left(1 + \frac{1}{1 - b}\right) \cdot \sum_{i = 1}^n (4i)^{-2} \right) \cdot f(OPT)
	\enspace.
\]
Choosing $b = 1 - e^{-1}$, we get:
\begin{align*}
	F(y)
	={} &
	F(y \wedge (1 \cdot \NN))
	\geq
	F(y \wedge ((1 - e^{-1}) \cdot \NN)) - \ee e^{-1} \cdot \sum_{i = 1}^n (4i)^{-2} \cdot f(OPT)\\
	\geq{} &
	e^{-1} \left(1 - 4\ee \cdot \sum_{i = 1}^n (4i)^{-2} \right) \cdot f(OPT) - \ee e^{-1} \cdot \sum_{i = 1}^n (4i)^{-2} \cdot f(OPT)\\
	\geq{} &
	\left(e^{-1} - 4\ee \cdot \sum_{i = 1}^n (4i)^{-2} \right) \cdot f(OPT)
	\enspace.
\end{align*}

The corollary now follows by observing that:
\[
	4 \cdot \sum_{i = 1}^n (4i)^{-2}
	\geq
	\frac{1}{4} + 4 \cdot \int_1^n \frac{dx}{(4x)^2}
	=
	\frac{1}{4} - \left. \frac{1}{4x} \right|_1^n
	<
	\frac{1}{4} + \frac{1}{4}
	=
	\frac{1}{2}
	\enspace.
	\qedhere
\]
\end{proof}

We are now ready to prove the competitive ratio of Algorithm~\ref{alg:IncreaseFromZero_Polynomial}, and complete the proof of Theorem~\ref{th:unconsrained_positive}.

\begin{lemma}
Algorithm~\ref{alg:IncreaseFromZero_Polynomial} has a competitive ratio of at least $e^{-1} - \ee$.
\end{lemma}
\begin{proof}
By Observation~\ref{obs:y_alg_equivalent_Polynomial}, the expected value of the set produced by Algorithm~\ref{alg:IncreaseFromZero_Polynomial} is $\mathbb{E}[F(y(\VV))]$. By the law of total expectation:
\[
	\mathbb{E}[F(y(\VV))]
	\geq
	\Pr[\EE] \cdot \mathbb{E}[F(y(\VV)) \mid \EE]
	\geq
	(1 - \ee/2) \cdot (e^{-1} - \ee/2) \cdot f(OPT)
	\geq
	(e^{-1} - \ee) \cdot f(OPT)
	\enspace,
\]
where the second inequality holds by Corollaries~\ref{cor:EE_probability} and~\ref{cor:unconstrained_ratio_Polynomial}.
\end{proof}}
\inArxiv{\section{A Polynomial Time Implementation of Algorithm~\ref{alg:GreedyWithThreshold} for the Objective \texorpdfstring{$g(S)$}{g(S)}} \label{app:uniform_polynomial}

In this section we present a polynomial time implementation of Algorithm~\ref{alg:GreedyWithThreshold} that can be used to get a set $S$ obeying $g(S) \geq (56/627 - \ee) \cdot f(OPT)$, for any constant $\ee > 0$, and thus, together with the rounding described in Section~\ref{sc:uniform}, proves the second part of Theorem~\ref{th:cardinality_positive}. The polynomial time implementation is given as Algorithm~\ref{alg:GreedyWithThreshold_Polynomial}.

\begin{algorithm*}[h!t]
\caption{\textsf{Greedy with Theshold for Non-monotone Functions - Polynomial Time Implementation}$(c, p)$} \label{alg:GreedyWithThreshold_Polynomial}
Let $S_0 \gets \varnothing$.\\
\ForEach{element $u_i$ revealed}
{
	\If{$i \leq pk$}
	{
		Let $S_i \gets S_{i - 1} + u_i$.\\
	}
	\Else
	{
		\For{each $u \in S_{i - 1}$}
		{
			Approximate
			\begin{align*}
				m_{u, i} ={} & g(S_{i - 1} + u_i - u'_i) - (1 + c/(pk)) \cdot g(S_{i - 1}) \\={} & \mathbb{E}[f(\RSet(p^{-1}(S_{i - 1} + u_i - u'_i))) - (1 + c/(pk)) \cdot f(\RSet(p^{-1} \cdot S_{i - 1}))]
			\end{align*}
			using $\ell = \lceil 50000 \ee^{-2}i^6(p + c)^2 \ln[2 \ee^{-1} pk(2i)^2] \rceil$ samples.\\
		}
	
		\BlankLine
	
		Let $u'_i$ be the element of $S_{i - 1}$ maximizing $m_{u'_i, i}$.\\
		\If{$m_{u'_i, i} \geq 0$}
		{
			Let $S_i \gets S_{i - 1} + u_i - u'_i$.
		}
		\Else
		{
			Let $S_i \gets S_{i - 1}$.
		}
	}
}
\end{algorithm*}

\begin{lemma}
Every value $m_{u, i}$ calculated by Algorithm~\ref{alg:GreedyWithThreshold_Polynomial} obeys:
\[
	\Pr[|m_{u, i} - g(S_{i - 1} + u_i - u'_i) + (1 + c/(pk)) \cdot g(S_{i - 1})| \leq \ee(4i)^{-3} \cdot f(OPT)] \leq \ee p^{-1}k^{-1}(2i)^{-2}
	\enspace.
\]
\end{lemma}
\begin{proof}
The value $m_{u, i}$ is calculated by averaging $\ell$ samples. For every set $S$ of size at most $pk$, submodularity implies $f(S) \leq p \cdot f(OPT)$. Hence, each sample must belong to the range $[-(p + c/k) \cdot f(OPT), p \cdot f(OPT)]$. Thus, by Lemma~\ref{lem:concentration_bound},
\begin{align*}
	\Pr[|m_{u, i} - g(S_{i - 1} + u_i - u'_i) + (1 + c/(pk)) \cdot g(S_{i - 1})| &\leq \ee(4i)^{-3} \cdot f(OPT)]
	\leq
	2e^{-\frac{\ell[\ee(4i)^{-3} / (p + c)]^2}{12}}\\
	\leq{} &
	2e^{-\ln[2 \ee^{-1} pk(2i)^2]}
	=
	\ee p^{-1}k^{-1}(2i)^{-2}
	\enspace.
	\qedhere
\end{align*}
\end{proof}

\begin{corollary} \label{cor:approximations_right}
With probability at least $1 - \ee/2$, all the values $m_{u, i}$ calculated by Algorithm~\ref{alg:GreedyWithThreshold_Polynomial} are accurate approximations up to an error of $\ee(4i)^{-3} \cdot f(OPT)$.
\end{corollary}
\begin{proof}
For every $1 \leq i \leq n$, Algorithm~\ref{alg:GreedyWithThreshold_Polynomial} approximates at most $pk$ different $m_{u, i}$ values. Hence, by the union bound, the probability that all these approximations are correct up to an error of $\ee(4i)^{-3} \cdot f(OPT)$ is at least:
\[
	1 - pk \cdot \sum_{i = 1}^n(\ee p^{-1}k^{-1}(2i)^{-2})
	=
	1 - \ee \cdot \sum_{i = 1}^n (2i)^{-2}
	\geq
	1 - \frac{\ee}{4} - \ee \cdot \int_1^n \frac{dx}{4x}
	=
	1 - \frac{\ee}{4} + \left. \frac{\ee}{4x} \right|_1^n
	>
	1 - \frac{\ee}{2}
	\enspace.
	\qedhere
\]
\end{proof}

Let $\EE$ be the event that Corollary~\ref{cor:approximations_right} holds. As in Section~\ref{sc:uniform}, for every $0 \leq i \leq n$ we define $A_i = \cup_{j = 1}^i S_j$. The following two lemmata correspond to Lemmata~\ref{le:maximum_better_than_average_random} and~\ref{le:S_fraction_of_A_random}.

\begin{lemma} \label{le:maximum_better_than_average_random_Polynomial}
Assuming $\EE$ holds, then for every $pk < i \leq n$, $g(S_{i - 1} + u_i - u'_i) - g(S_{i - 1}) \geq g(u_i \mid A_{i - 1}) - g(S_{i - 1}) / (pk) - 2\ee(4i)^{-3} \cdot f(OPT)$.
\end{lemma}
\begin{proof}
By the choice of $u'_i$ and an averaging argument,
\begin{align*}
	g(&S_{i - 1} + u_i - u'_i) - g(S_{i - 1}) 
	\geq
	m_{u'_i, i} + c/(pk) \cdot g(S_{i - 1}) - \ee(4i)^{-3} \cdot f(OPT)\\
	\geq{} &
	\frac{\sum_{u' \in S_{i - 1}} [m_{u', i} + c/(pk) \cdot g(S_{i - 1}) - \ee(4i)^{-3} \cdot f(OPT)]}{pk}\\
	\geq{} &
	\frac{\sum_{u' \in S_{i - 1}} [g(S_{i - 1} + u_i - u') - g(S_{i - 1}) - 2\ee(4i)^{-3} \cdot f(OPT)]}{pk}\\
	\geq{} &
	\frac{\sum_{u' \in S_{i - 1}} [g(S_{i - 1} + u_i) - g(S_{i - 1})]}{pk} + \frac{\sum_{u' \in S_{i - 1}} [g(S_{i - 1} - u') - g(S_{i - 1})]}{pk} - 2\ee(4i)^{-3} \cdot f(OPT)
	\enspace,
\end{align*}
where the last inequality follows from submodularity. The first term on the rightmost hand side is $g(u_i \mid S_{i - 1})$, which can be lower bounded by $g(u_i \mid A_{i - 1})$  because $S_{i - 1} \subseteq A_{i - 1}$. The second term on the rightmost hand side can be lower bounded using the submodularity and non-negativity of $g$ as follows:
\begin{align*}
	\frac{\sum_{u' \in S_{i - 1}} [g(S_{i - 1} - u') - g(S_{i - 1})]}{pk}
	\geq{} &
	\frac{g(\varnothing) - g(S_{i - 1})}{pk}\inConference{\\}
	\geq\inConference{{} &}
	- \frac{g(S_{i - 1})}{pk}
	\enspace.
	\qedhere
\end{align*}
\end{proof}

The following lemma relates $S_i$ and $A_i$.

\begin{lemma} \label{le:S_fraction_of_A_random_Polynomial}
Assuming $\EE$ holds, then for every $0 \leq i \leq n$, $g(S_i) \geq \frac{c}{c + 1} \cdot g(A_i) - 2\ee \cdot \sum_{j = 1}^i (4j)^{-3} \cdot f(OPT) \geq g(A_i) - \ee/4 \cdot f(OPT)$.
\end{lemma}
\begin{proof}
The second inequality holds since $\sum_{j = 1}^i (4j)^{-3} = 0$ for $i = 0$ and for $i > 0$:
\[
	2 \cdot \sum_{j = 1}^i (4j)^{-3}
	\leq
	\frac{1}{32} + \int_1^i \frac{dx}{32x^3}
	=
	\frac{1}{32} - \left. \frac{1}{64x^2}\right|_1^i
	<
	\frac{1}{32} + \frac{1}{64}
	<
	\frac{1}{4}
	\enspace.
\]
Hence, in the rest of the proof we concentrate on proving the first inequality.

For $i = 0$, the inequality clearly holds since $g(A_0) = g(\varnothing) = g(S_0)$. Thus, it is enough to prove that for every $1 \leq i \leq n$, $g(S_i) - g(S_{i - 1}) \geq \frac{c}{c + 1} \cdot [g(A_i) - g(A_{i - 1})] - 2\ee(4i)^{-3} \cdot f(OPT)$. If the algorithm does not accept $u_i$ then $S_i = S_{i - 1}$ and  $A_i = A_{i - 1}$, and the claim is trivial. Hence, we only need to consider the case in which the algorithm accepts element $u_i$.

If $i \leq k$ then $g(S_i) - g(S_{i - 1}) = g(u_i \mid S_{i - 1}) \geq g(u_i \mid A_{i - 1}) = g(A_i) - g(A_{i - 1})$, where the inequality follows from submodularity since $S_{i - 1} \subseteq A_{i - 1}$. If $i > k$, then it is possible to lower bound $g(S_i) - g(S_{i - 1}) = g(S_{i - 1} + u_i - u'_i) - g(S_{i - 1})$ in two ways. First, a lower bound of $g(u_i \mid A_{i - 1}) - g(S_{i - 1}) / k - 2\ee(4i)^{-3} \cdot f(OPT)$ is given by Lemma~\ref{le:maximum_better_than_average}. The second lower bound states that, since the algorithm accepts $u_i$,
\begin{align*}
	g(S_{i - 1} + u_i - u'_i) - g(S_{i - 1})
	\geq{} &
	m_{u'_i, i} + c \cdot g(S_{i - 1}) / (pk) - \ee(4i)^{-3} \cdot f(OPT)\\
	\geq{} &
	c \cdot g(S_{i - 1}) / (pk) - \ee(4i)^{-3} \cdot f(OPT)
	\enspace.
\end{align*}
Using both lower bounds, we get:
\begin{align*}
	g(S_i) - g(S_{i - 1})
	\geq{} &
	\max\left\{g(u_i \mid A_{i - 1}) - \frac{g(S_{i - 1})}{pk}, \frac{c \cdot g(S_{i - 1})}{pk}\right\}- 2\ee(4i)^{-3} \cdot f(OPT)\\
	\geq{} &
	\frac{c \cdot [g(u_i \mid A_{i - 1}) - g(S_{i - 1})/(pk)] + 1 \cdot [c \cdot g(S_{i - 1})/(pk)]}{c + 1} - 2\ee(4i)^{-3} \cdot f(OPT)\\
	\geq{} &
	\frac{c}{c + 1} \cdot g(u_i \mid A_{i - 1}) - 2\ee(4i)^{-3} \cdot f(OPT)\\
	={} &
	\frac{c}{c + 1} \cdot [g(A_i) - g(A_{i - 1})] - 2\ee(4i)^{-3} \cdot f(OPT)
	\enspace.
	\qedhere
\end{align*}
\end{proof}

We also need to show that $g(S_i)$ is roughly monotone as a function of $i$.
\begin{lemma} \label{lem:monotonicity}
Assuming $\EE$ holds, then for $pk < i \leq n$, $g(S_i) - g(S_{i - 1}) \geq -\ee(4i)^{-3} \cdot f(OPT)$.
\end{lemma}
\begin{proof}
If $S_i = S_{i - 1}$ then the claim is trivial. Otherwise, we must have $m_{u'_i, i} \geq 0$, and thus,
\[
	g(S_i) - g(S_{i - 1})
	\geq
	\frac{c \cdot g(S_i)}{pk} - \ee(4i)^{-3} \cdot f(OPT)
	\geq
	- \ee(4i)^{-3} \cdot f(OPT)
	\enspace.
	\qedhere
\]
\end{proof}
\begin{corollary} \label{cor:monotonicity}
Assuming $\EE$ holds, then for $pk \leq i \leq n$, $g(S_i) \geq g(S_n) - \ee(8i)^{2} \cdot f(OPT)$.
\end{corollary}
\begin{proof}
By Lemma~\ref{lem:monotonicity},
\begin{align*}
	g(S_n) - g(S_i)
	\geq{} &
	-\ee \cdot \sum_{j = i + 1}^n (4i)^{-3} \cdot f(OPT)
	\geq
	-\ee \cdot \int_i^n \frac{dx}{64x^3} \cdot f(OPT)\\
	={} &
	\ee \cdot \left. \frac{1}{128x^2} \right|_i^n \cdot f(OPT)
	\geq
	- \frac{\ee}{128i^2} \cdot f(OPT)
	\geq
	- \frac{\ee}{(8i)^2} \cdot f(OPT)
	\enspace.
	\qedhere
\end{align*}
\end{proof}

Using the above claims, it is possible to get a guarantee on $g(S_n)$ assuming $\EE$ holds. The following corollary corresponds to Corollary~\ref{co:S_n_guarantee}.

\begin{corollary} \label{co:S_n_guarantee_Polynomial}
Assuming $\EE$ holds, Algorithm~\ref{alg:GreedyWithThreshold_Polynomial} produces a set $S_n$ such that:
\[
	g(S_n) \geq \left(\frac{cp^{-1}(1 - p^{-1})}{(c + 1)(1 + p^{-1}c)} - \frac{\ee}{2}\right) \cdot f(OPT)
	\enspace.
\]
\end{corollary}
\begin{proof}
Consider an element $u_i \in OPT \setminus A_i$. Since $u_i$ was rejected by Algorithm~\ref{alg:GreedyWithThreshold_Polynomial}, we must have $m_{u'_i, i} < 0$, which implies:
\begin{align*}
	\frac{c \cdot g(S_{i - 1})}{pk} + \ee(4i)^{-3} \cdot f(OPT)
	>{} &
	g(S_{i - 1} + u_i - u'_i) - g(S_{i - 1})\\
	\geq{} &
	g(u_i \mid A_{i - 1}) - \frac{g(S_{i - 1})}{pk} - 2\ee(4i)^{-3} \cdot f(OPT)
	\enspace,
\end{align*}
where the second inequality follows from Lemma~\ref{le:maximum_better_than_average_random_Polynomial}. Rearranging yields:
\[
	g(u_i \mid A_{i - 1})
	<
	\frac{c+1}{pk} \cdot g(S_{i - 1}) + 3\ee(4i)^{-3} \cdot f(OPT)
	\leq
	\frac{c+1}{pk} \cdot g(S_n) + 4\ee(8i)^{-2} \cdot f(OPT)
	\enspace,
\]
where the last inequality uses Corollary~\ref{cor:monotonicity}. By the submodularity of $g$ and Lemma~\ref{le:S_fraction_of_A_random_Polynomial}, we now get:
\begin{align*}
	g(A_n \cup OPT)
	\leq{} &
	g(A_n) + \sum_{u \in OPT \setminus A_n} g(u_i \mid A_n)\\
	<{} &
	g(A_n) + \sum_{u \in OPT \setminus A_n} \left(\frac{c + 1}{pk} \cdot g(S_n)\right) + 4\ee \cdot \sum_{i = 1}^n (8i)^{-2} \cdot f(OPT)\\
	\leq{} &
	g(A_n) + p^{-1}(c + 1) \cdot g(S_n) + 4\ee \cdot \sum_{i = 1}^n (8i)^{-2} \cdot f(OPT)\\
	\leq{} &
	\frac{(c + 1)(1 + p^{-1}c)}{c} \cdot g(S_n) + \frac{c + 1}{c} \cdot \frac{\ee}{4} \cdot f(OPT) + 4\ee \cdot \sum_{i = 1}^n (8i)^{-2} \cdot f(OPT)
	\enspace.
\end{align*}

Observe now that:
\[
	4\ee \cdot \sum_{i = 1}^n (8i)^{-2} \cdot f(OPT)
	\leq
	\frac{\ee}{16} + \int_1^n \frac{dx}{16x^2}
	=
	\frac{\ee}{16} - \left. \frac{1}{16x}\right|_1^n
	\leq
	\frac{\ee}{16} + \frac{1}{16}
	<
	\frac{\ee}{4}
	\enspace,
\]
and, by Lemma~\ref{le:sampling_lowerbound},
\begin{align*}
	g(A_n \cup OPT)
	={} &
	F(p^{-1} \cdot OPT + p^{-1} \cdot (A_n \setminus OPT))\inConference{\\}
	\geq\inConference{{} &}
	p^{-1}(1 - p^{-1}) \cdot f(OPT)
	\enspace.
\end{align*}

Combining the three last inequalities yields:
\[
	\left[p^{-1}(1 - p^{-1}) - \frac{c + 1}{c} \cdot \frac{\ee}{2}\right] \cdot f(OPT)
	\leq
	\frac{(c + 1)(1 + p^{-1}c)}{c} \cdot g(S_n)
	\enspace.
\]
The corollary now follow by dividing the last inequality by $\frac{(c + 1)(1 + p^{-1}c)}{c}$ and observing that $1 + p^{-1}c \geq 1$.
\end{proof}

\begin{corollary}
Algorithm~\ref{alg:GreedyWithThreshold_Polynomial} produces a set $S_n$ such that:
\[
	\mathbb{E}[g(S_n)] \geq \left(\frac{cp^{-1}(1 - p^{-1})}{(c + 1)(1 + p^{-1}c)} - \ee\right) \cdot f(OPT)
	\enspace.
\]
Hence, for $c = 7/4$ and $p = 3$, $g(S_n) \geq (56/627 - \ee) \cdot f(OPT)$.
\end{corollary}
\begin{proof}
By the law of total expectation,
\begin{align*}
	\mathbb{E}[g(S_n)]
	\geq{} &
	\Pr[\EE] \cdot \mathbb{E}[g(S_n) \mid \EE]
	\geq
	(1 - \ee/2) \cdot \left(\frac{cp^{-1}(1 - p^{-1})}{(c + 1)(1 + p^{-1}c)} - \frac{\ee}{2}\right) \cdot f(OPT)\\
	\geq{} &
	\left(\frac{cp^{-1}(1 - p^{-1})}{(c + 1)(1 + p^{-1}c)} - \ee\right) \cdot f(OPT)
	\enspace,
\end{align*}
where the second inequality holds by Corollaries~\ref{cor:approximations_right} and~\ref{co:S_n_guarantee_Polynomial}.
\end{proof}}

\end{document}